\newcommand{\CLASSINPUTinnersidemargin}{1.015in}
\newcommand{\CLASSINPUToutersidemargin}{1.015in}
\documentclass[12pt,onecolumn]{IEEEtran}
\normalsize
\usepackage{algorithm,amsbsy,amsmath,amssymb,epsfig,bbm,mathrsfs,multirow,amsthm}
\usepackage{algpseudocode}
\usepackage{graphicx,caption,subcaption}

\usepackage[hidelinks]{hyperref}
\renewcommand{\baselinestretch}{1.54}
\DeclareMathAlphabet{\mathpzc}{OT1}{pzc}{m}{it}

\newtheorem{proposition}{Proposition} 

\newtheorem{theorem}{\textbf{\textsc{Theorem}}}

\begin{document}
%
\renewcommand{\baselinestretch}{1}
\title{Optimal and Low-Complexity Dynamic Spectrum Access for RF-Powered Ambient Backscatter System with Online Reinforcement Learning}

\author{Nguyen~Van~Huynh, Dinh~Thai~Hoang, Diep~N.~Nguyen, Eryk~Dutkiewicz,\\
	Dusit~Niyato, and~Ping~Wang
\thanks{Preliminary results in this paper will be presented at the IEEE Globecom Conference, 2018~\cite{globecom}}}

\maketitle

\renewcommand{\baselinestretch}{1}	
\begin{abstract}
Ambient backscatter has been introduced with a wide range of applications for low power wireless communications. In this article, we propose an optimal and low-complexity dynamic spectrum access framework for RF-powered ambient backscatter system. In this system, the secondary transmitter not only harvests energy from ambient signals (from incumbent users), but also backscatters these signals to its receiver for data transmission. Under the dynamics of the ambient signals, we first adopt the Markov decision process (MDP) framework to obtain the optimal policy for the secondary transmitter, aiming to maximize the system throughput. However, the MDP-based optimization requires complete knowledge of environment parameters, e.g., the probability of a channel to be idle and the probability of a successful packet transmission, that may not be practical to obtain. To cope with such incomplete knowledge of the environment, we develop a low-complexity online reinforcement learning algorithm that allows the secondary transmitter to ``learn" from its decisions and then attain the optimal policy. Simulation results show that the proposed learning algorithm not only efficiently deals with the dynamics of the environment, but also improves the average throughput up to 50\% and reduces the blocking probability and delay up to 80\% compared with conventional methods.
\end{abstract}

\begin{IEEEkeywords}
Ambient backscatter, RF energy harvesting, dynamic spectrum access, Markov decision process, reinforcement learning.
\end{IEEEkeywords}

\IEEEpeerreviewmaketitle

\renewcommand{\baselinestretch}{1.533}	
\section{Introduction}
\label{sec:Introduction}

Dynamic spectrum access (DSA) has been considered as a promising solution to improve the utilization of radio spectrum~\cite{DSA}. As DSA standard frameworks, the Federal Communications Commission and the European Telecommunications Standardization Institute have recently proposed Spectrum Access Systems (SAS) and Licensed Shared Access (LSA) respectively~\cite{SAS_LSA}. In both SAS and LSA, spectrum users are prioritized at different levels/tiers (e.g., there are three types of users with a decreasing order of priority: Incumbent Users (IUs), Priority Access Licensees (PALs), and General Authorized Access (GAAs)). Without loss of generality, in this work, we refer users with higher priority as IUs and users with lower priority as secondary users (SUs). DSA harvests under-utilized spectrum chunks by allowing an SU to dynamically access (temporarily) idle spectrum bands/whitespaces to transmit data.

For low-power communications users in DSA (e.g., IoT applications), recent advances in radio frequency (RF) energy harvesting allow SUs to further leverage/exploit the IUs' signals/bands even while IUs are active. Specifically, with RF-energy harvesting capability, an SU transmitter can harvest/capture energy from the incumbent signals that are transmitted from IUs, e.g., base stations and TV towers. Later, when the incumbent channel is idle, the SU can use the harvested energy to transmit its data. This mechanism is also known as the harvest-then-transmit (HTT) technique~\cite{Huynh2017Survey} that can improve both the spectrum utilization and the energy efficiency. However, the SU's system performance under HTT strongly depends on the amount of harvested energy. Intuitively, if the incumbent channel is mostly idle, the amount of harvested energy is insignificant, and thus the SU may not have sufficient energy to transmit data. In the case when the incumbent channel remains busy for a long period, the SU may not be able to use all the harvested energy to transmit data due to the transmission power regulation and the limited transmission time.

Given the above, we introduce a novel framework that employs ambient backscatter communications to further improve the spectrum utilization of RF-powered DSA systems. The ambient backscatter technology has been emerging recently as an enabler for ubiquitous communications~\cite{LiuAmbient2013}-\cite{Kellogg2014Wi-fi}. Unlike conventional backscatter communication systems, i.e., monostatic and bistatic backscatter~\cite{Huynh2017Survey} that require dedicated RF sources, in an ambient backscatter communication system, wireless devices can communicate just by reflecting RF signals from ambient RF sources, e.g., TV towers, cellular base stations, and Wi-Fi APs. Thus, this technique not only reduces deployment and maintenance costs, but also supports device-to-device communications with a very small environmental footprint. By integrating the ambient backscatter technique into an RF-powered DSA system, the secondary transmitter (ST){\footnote{The principle as well as the circuit design of the ST will be discussed in Section~\ref{sec:sysmodel}.}} can transmit data to its secondary receiver (SR) by backscattering the incumbent signals when the IU is active. Hence, the SU with ambient backscatter will have more options to transmit data with ultra-low energy consumption, further improving the spectrum utilization while causing no harmful interference to IUs \cite{LiuAmbient2013}. Note that, with the recent advances in coding and detection mechanisms~\cite{Huynh2017Survey}, the transmission range of the ambient backscatter technique can be extended up to 100 meters, making a very promising solution for the next generation of low-power wireless communications systems.

However, RF signals from IUs, e.g., TV or radar towers, are often highly dynamic and even unknown to the SUs due to RF source activities and the locations of the SUs. Furthermore, ambient backscatter and RF energy harvesting can not be efficiently performed on the same wireless device simultaneously~\cite{Zhang2014Enabling}. A critical challenge to RF-powered ambient backscatter DSA systems is how to efficiently tradeoff between backscattering RF signals (to transmit data) and harvesting energy from RF signals (to sustain the internal operation for the SU) under the dynamic of the ambient signals. In addition, the low-power SUs are intrinsically limited in computing and energy. This fact calls for efficient yet lightweight solutions.

This work aims to provide an optimal and efficient DSA solution that guides the ambient backscatter ST to whether stay idle, backscatter signals, harvest energy, or actively transmit data, to maximize its throughput
(based on its current observations, i.e., channel state, the energy level, and data buffer status). 
In particular, we first develop a Markov decision process framework together with linear programming technique to obtain the dynamic optimal policy for the ST when all environment information is given in advance. We then propose a low-complexity online learning algorithm to help the ST make the optimal decisions when the environment parameters, e.g., channel state, the successful data transmission probabilities, are not available. The simulations demonstrate that the proposed solutions always achieve the best performance compared with other existing methods. Furthermore, the proposed learning algorithm with incomplete environment parameters can closely attain the performance of the MDP optimization with complete information.

\section{Related Work And Main Contributions}
\label{sec:relatedwork}
\subsection{Related Work}

As aforementioned, the performance of a RF-powered DSA system strongly depends on the amount of harvested energy and/or battery capacity of the SUs~\cite{hasan2011}. Various works in the literature study the joint optimization of energy harvesting and data transmission to maximize the spectrum utilization, e.g., \cite{Lee2013Opportunistic}-\cite{Sakr2015Cognitive}. In particular, the authors of \cite{Derrick2016Multi} propose a non-convex multi-objective optimization problem to maximize the energy harvesting efficiency and minimize the total transmit power as well as the interference power leakage-to-transmit power ratio for a DSA. In \cite{Sakr2015Cognitive}, the authors consider device-to-device (D2D) communications in a cellular network. In this network, the D2D transmitters harvest energy from ambient RF signals and use the uplink or downlink channel to actively communicate with the corresponding receivers.

However, all current solutions for RF-powered DSA systems encounter a common limitation when the incumbent channel is mostly busy. In such a case, the throughput of the secondary system is low as the ST hardly has opportunities to access the IUs' channel to transmit. The ambient backscatter technique has recently emerged as a promising solution to address this problem. The ambient backscatter technique is particularly appropriate for implementation in RF-powered DSA systems due to the following reasons. First, ambient backscatter circuits are small with low-energy consumption~\cite{Penichet2016Do}-\cite{Liu2017Coding}, while they can share the same antenna in RF-powered wireless devices. Second, similar to RF-powered DSA systems, the ambient backscatter technique also utilizes incumbent signals as the resource to transmit data, thereby maximizing the spectrum utilization. Third, the ambient backscatter technique can transmit data without requiring decoding incumbent signals, thereby lowering the complexity of the secondary systems. However, when the ambient backscatter technique is integrated into RF-powered DSA wireless devices, how to tradeoff between the HTT and backscatter activities in order to maximize network throughput of the ST is a major challenge.

The optimal time tradeoff between the HTT and backscatter activities is investigated in \cite{Hoang2017Ambient}. In this work, the authors prove that there always exists the globally optimal time tradeoff. This implies that the integration of the ambient backscatter technique into RF-powered DSA systems always achieves the overall transmission rate higher than that of using either the ambient backscatter communication or the HTT scheme individually. In~\cite{Kim2017Hybrid}, the authors propose a hybrid backscatter communication system to improve transmission range and bitrate. Different from~\cite{Hoang2017Ambient}, this system adopts a dual/hybrid mode operation of bistatic backscatter and ambient backscatter depending on indoor and outdoor zones, respectively. Through numerical results, the authors show that the proposed hybrid communication can significantly increase the throughput and coverage of the system. Nevertheless, these solutions require complete knowledge of environment parameters to formulate the optimization problem. Alternatively, a stochastic geometry model is used in~\cite{Han2017Wirelessly} to derive the success transmission probability together with the network transmission capacity. To improve the average throughput and the coverage of backscatter networks, the sensing-pricing-transmitting policy adaptation problem for the STs is investigated in~\cite{Wang2018Stackelberg} through using a Stackeberg game model. Nonetheless, the game model does not deal with the dynamics of the environment and may be infeasible to deploy in the ST which is a power-constrained device.

All aforementioned and other related work in the literature have not accounted for the dynamics of the incumbent signals. In this work, we capture the dynamics of IUs using the MDP framework to obtain optimal DSA decisions for the ST. However, the MDP optimization requires complete environment information to derive the optimal policy that may not be practical in dynamic systems. Moreover, to deal with such incomplete information scenarios or environment uncertainties, the MDP optimization becomes more computationally demanding, especially for large-scale systems. To address these shortcomings, we propose a low-complexity online reinforcement learning algorithm.

\subsection{Contributions}
The major contributions of this paper are as follows.
\begin{itemize}

	\item We develop a dynamic optimization framework based on MDP to obtain the optimal policy for the ST in the RF-powered ambient backscatter communications system. This policy allows the ST to make optimal decisions to maximize its long-term average throughput under the dynamics of incumbent channel state, data, and energy demands. 
	
	\item To find the optimal policy for the ST, we first construct the transition probability matrix and formulate the optimization problem. Then, we use linear programming~\cite{Puterman_1994_Book} to obtain the optimal policy.
	
	\item To deal with the incomplete information and high-complexity of traditional methods, we propose a low-complexity online reinforcement learning algorithm that allows the ST to obtain the optimal policy through learning from its decisions. The proposed learning algorithm is especially important for RF-powered ambient backscatter IoT devices that have limited power and computing resources but have to deal with the dynamics of the surrounding environment.
	 
	\item Finally, we perform extensive performance evaluation with the aims of not only demonstrating the efficiency of the proposed solutions, but also providing insightful guidance on the implementation of RF-powered ambient backscatter DSA systems.
\end{itemize}
Section~\ref{sec:sysmodel} describes the system model. Section~\ref{sec:markov} presents the MDP framework together with the linear programming solution. The low-complexity online learning algorithm is developed in Section~\ref{sec:LA}. Finally, evaluation results are discussed in Section~\ref{sec:PE} and conclusions of the paper are drawn in Section~\ref{sec:conclusion}.
\section{System Model}
\label{sec:sysmodel}

\subsection{System Model}

In this work, we consider a DSA system in which the secondary system coexists with the incumbent system as illustrated in Fig.~\ref{Fig.system_model}. The secondary system consists of an ST and an SR, and the ST will opportunistically transmit data to the SR in an overlay fashion. The ST is equipped with ambient backscatter and RF energy harvesting circuits. When the incumbent transmitter transmits data to its receiver, i.e., the incumbent channel is busy, the ST can either backscatter the incumbent signals to transmit data or harvest energy from the signals to store the energy in its energy storage as shown in Fig.~\ref{Fig.system_model}(a). In contrast, when the channel is idle, the ST can use its harvested energy to actively transmit data to its SR as shown in Fig.~\ref{Fig.system_model}(b).

\begin{figure}[tbh]
	\captionsetup{singlelinecheck=off}
	\centering
	\includegraphics[scale=0.28]{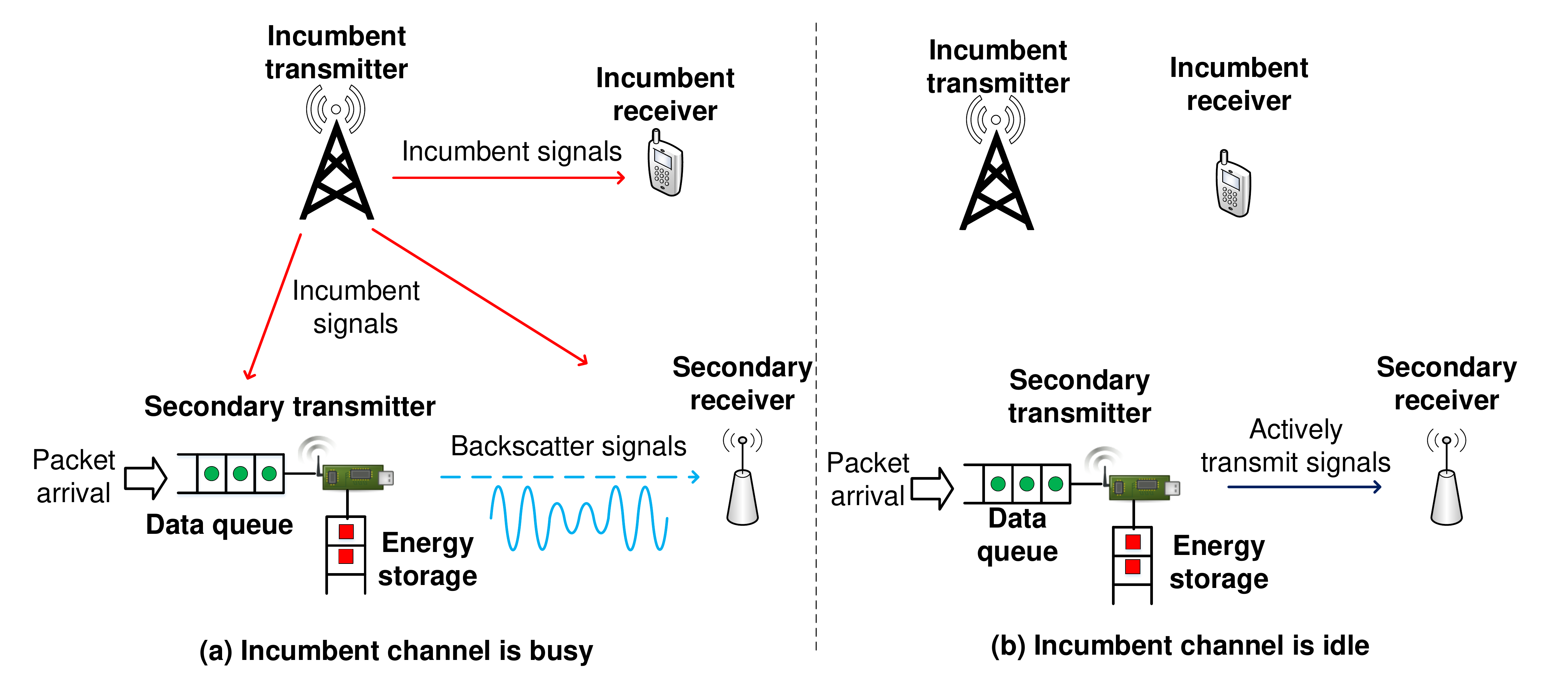}
	\caption{DSA RF-powered ambient backscatter system model.}
	\label{Fig.system_model}
\end{figure}

The maximum data queue size and the energy storage capacity are denoted by $D$ and $E$, respectively. In each time slot, the probability of a packet arriving at the data queue is denoted by $\alpha$. We denote the probability of the incumbent channel being idle by $\eta$. When the channel is busy and the ST performs backscattering to transmit data, the ST can transmit $d_\mathrm{b}$ data units successfully with probability $\beta$. This transmission is referred to as the backscatter mode. 
When the channel is busy and if the ST chooses to harvest energy, it can harvest $e_\mathrm{h}$ units of energy successfully with probability $\gamma$. When the channel becomes idle, the ST can use $e_\mathrm{t}$ units of energy to actively transmit $d_\mathrm{t}$ data units to its receiver, and $\sigma$ denotes the successful data transmission probability. This process is also known as harvest-then-transmit (HTT) mode~\cite{park2013}. Note that we consider only one ST as it is a typical setting for backscatter communications systems~\cite{Parks2014Turbocharging},~\cite{Vou2016Could},~\cite{Kimionis2012Bistatic}. Nevertheless, multiple STs can be supported through the channel selection that allows the STs to operate on different incumbent channels to avoid collision and complex signaling as in random access and TDMA, respectively. This case can be extended straightforward from the current system model.

\subsection{DSA RF-Powered Ambient Backscatter Circuit Diagram}

Fig.~\ref{Fig.circuit_diagram} shows a circuit diagram implemented at the ST and the SR in our considered RF-powered ambient backscatter DSA system. This circuit diagram has been adopted in many hardware designs in the literature \cite{Huynh2017Survey}, \cite{LiuAmbient2013}, \cite{Vou2016Could}, \cite{Kimionis2012Bistatic}. The ST consists of five main components, i.e., the controller, load modulator (for ambient backscatter process), energy harvester, active RF transmitter, a rechargeable battery, i.e., energy storage, and data buffer. The controller is responsible for controlling all the actions of the ST including making decisions and performing actions, e.g., stay idle, transmit data, harvest energy, and backscatter data. When the incumbent channel is busy, if the ST chooses to harvest energy, the ST will harvest energy from the incumbent RF signals by using the RF energy harvester and store the energy in the rechargeable battery. This energy will be used for transmitting data when the incumbent channel becomes idle through the active RF transmitter. In contrast, if the ST chooses to backscatter data, the ST will modulate the reflection of the ambient RF signals to send the data to the SR through the load modulator. To do so, the ST uses a switch which consists of a transistor connected to the antenna. The input of the ST is a stream of one and zero bits. When the input bit is zero, the transistor is off, and thus the ST is in the non-reflecting state. Otherwise, when the input bit is one, the transistor is on, and thus the ST is in the reflecting state. As such, the ST is able to transfer bits to the SR. Note that, in the backscatter mode, the ST can still harvest energy, but the amount of harvested energy is relatively small and just sufficient to supply for operations in the backscatter mode \cite{Huynh2017Survey}, \cite{LiuAmbient2013}.

\begin{figure*}[tbh]
	\centering
	\includegraphics[scale=0.22]{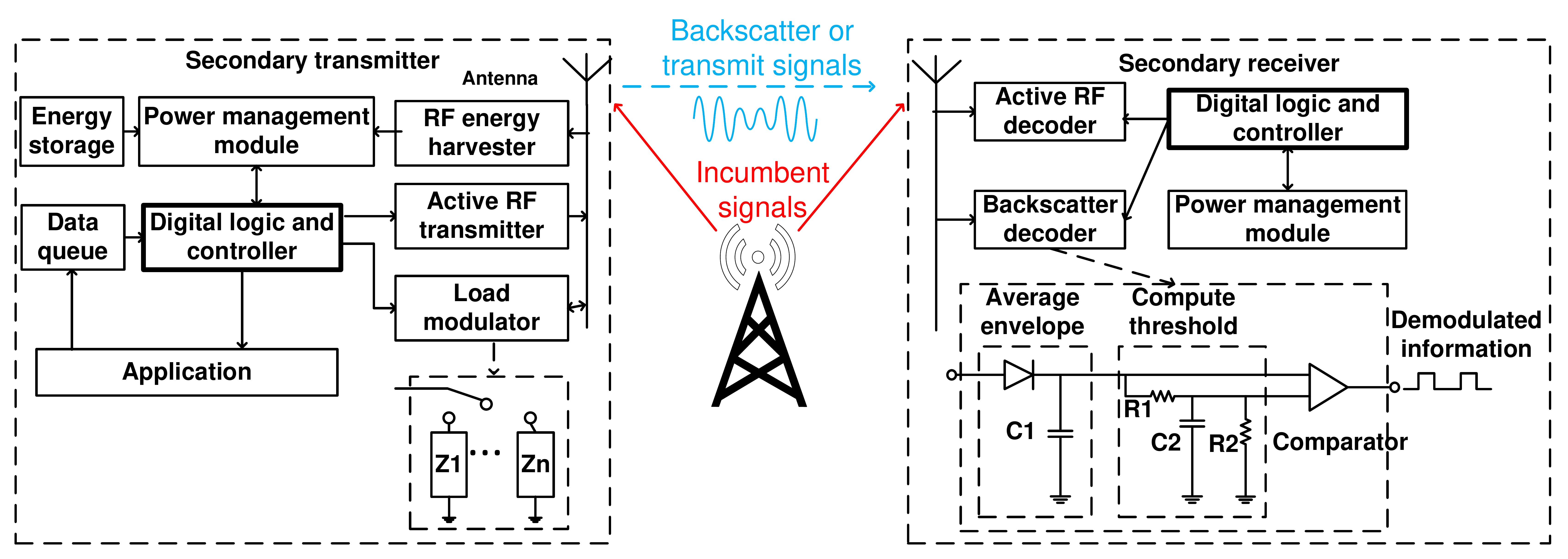}
	\caption{Circuit diagram of the DSA RF-powered ambient backscatter system.}
	\label{Fig.circuit_diagram}
\end{figure*}

The SR is equipped with the controller and power source. The controller takes responsibility for all the operations of the SR including selecting operation modes to extract the data sent from the ST. The active RF decoder is used to decode data when the ST actively transmits the data in the channel idle period. For the backscatter mode, the SR uses the backscatter decoder to extract the transmitted data. Specifically, to decode the data from the ST in the backscatter mode, the received signals are first smoothed by an envelope-averaging circuit. Then, the compute-threshold circuit is used to produce an output voltage between two levels, i.e., low and high. After that, the comparator compares the average envelope signals with a predefined threshold to generate output bits, i.e., $0$ or $1$.

\subsection{Tradeoff in DSA RF-Powered Ambient Backscatter System}
In the considered system, we consider two successive working periods of the incumbent transmitter, i.e., idle and busy. As mentioned, when the incumbent channel is busy, the ST can either backscatter signals to transmit data to the SR or harvest energy and store the harvested energy in the energy storage. When the incumbent channel is idle, the ST can actively transmit data to the SR by using the energy in the energy storage. This leads to a tradeoff problem between the HTT and backscatter process to maximize the network throughput. In particular, the ST needs to take an action, e.g., transmit data, harvest energy, or backscatter data, based on its current state, i.e., the joint channel, data queue, and energy storage states. To find the optimal policy for the ST, we adopt two methods as follows:
\begin{itemize}
	\item When the ST knows the environment parameters, the optimal policy is obtained by an optimization formulation based on the offline linear programming approach. The detail of this solution is given in Section~\ref{sec:markov}.
	\item If the environment parameters are not available in advance, we introduce an online learning algorithm to help the ST obtain the optimal policy through interaction processes with the environment. The details of the online reinforcement algorithm with low complexity are provided in Section~\ref{sec:LA}.
\end{itemize}

\section{Markov Decision Process Formulation}
\label{sec:markov}

In this section, we present the optimization problem based on the MDP framework to obtain an optimal policy for the ST. We first define the action and state spaces. Then, the transition probability matrix of the MDP is derived. Finally, the optimization formulation and performance measures are obtained.

\subsection{State Space and Action Space}
\label{state_action_spaces}

We define the state space of the ST as follows:
\begin{equation}
\begin{split}
\mathcal{S} =	\Big\{ ({\mathcal{C}}, {\mathcal{D}}, {\mathcal{E}} ): {\mathcal{C}} \in \{0,1\}; {\mathcal{D}} \in \{0,\ldots,d, \ldots, D\}; {\mathcal{E}} \in \{0,\ldots,e,\ldots, E \} \Big\},
\end{split}
\end{equation}
where $c \in \mathcal{C}$ represents the state of the incumbent channel, i.e., $c = 1$ when the incumbent channel is busy and $c = 0$ otherwise, $d \in {\mathcal{D}}$ and $e \in {\mathcal{E}}$ represent the number of data units in the data queue and the energy units in the energy storage of the ST, respectively. $D$ is the maximum data queue size, and $E$ is the maximum capacity of the energy storage. The state of the ST is then defined as a composite variable $s = (c,d,e) \in \mathcal{S}$, where $c$, $d$ and $e$ are the channel state, the data state, and the energy state, respectively. The ST can perform one of the four actions, i.e., stay idle, transmit data, harvest energy, and backscatter data. Then, the action space of the ST is defined by $\mathcal{A}	\triangleq \{a:a \in \{1,\ldots, 4\} \}$ , where
\begin{equation}
a 	=	\left\{	\begin{array}{ll}
1,	&	\mbox{when the ST stays idle},	\\
2,	&	\mbox{when the ST transmits data},	\\
3,	&	\mbox{when the ST harvests energy},	\\
4,	&	\mbox{when the ST backscatters data}.
\end{array}	\right.
\end{equation}
Additionally, the action space given states of the ST denoted by $\mathcal{A}_s$ comprises all possible actions that do not make a transition to an unreachable state. We then can express $\mathcal{A}_s$ as follows:
\begin{equation}
\mathcal{A}_s	=	\left\{	\begin{array}{ll}
\{1\},	&	\mbox{if $c=0$ and $d<d_\mathrm{t}$ OR $c=0$ and $e < e_\mathrm{t}$ OR $c=1$, $e=E$ and $d<d_\mathrm{b}$},\\
\{1,2\},	&	\mbox{if $c=0$, $d \geq d_\mathrm{t}$ and $e \geq e_\mathrm{t}$},	\\
\{3\},	&	\mbox{if $c=1$, $d < d_\mathrm{b}$ and $e<E$},	\\
\{4\},	&	\mbox{if $c=1$, $d \geq d_\mathrm{b}$ and $e=E$},	\\
\{3,4\},	&	\mbox{if $c=1$, $d \geq d_\mathrm{b}$ and $e<E$}.
\end{array}	\right.
\end{equation}

The first condition corresponds to the case when the incumbent channel is idle, and the number of data units in the data queue or the number of energy units in the energy storage is not enough for active transmission. This condition is also applied to the special case when the channel is busy, the number of data units in the data queue is not enough for backscattering, and the energy storage is full. The ST then can select only action $a = 1$, i.e., stay idle. The second condition corresponds to the case when the incumbent channel is idle and there are enough data and energy for active transmission. The third condition corresponds to the case when the incumbent channel is busy, the data in the data queue is not enough for backscattering, and the energy storage is not full. The ST, therefore, can select only action $a=3$, i.e., harvest energy. The fourth condition corresponds to the case when the incumbent channel is busy, there is enough data for backscattering, and the energy storage is full. In this case, the ST can only choose to backscatter. The fifth condition corresponds to the case when the incumbent channel is busy, there is enough data for backscattering, and the energy storage is not full.

\subsection{Transition Probability Matrix}
\label{subsec:trans}
We express the transition probability matrix given action $a \in \mathcal{A}$ as follows:

\begin{equation}
{\mathbf{P}}(a)	=	\left[	\begin{array}{cc}
\eta \mathbf{W}(a)	&	(1-\eta)\mathbf{W}(a)	\\
\eta\mathbf{W}(a)	&	(1-\eta)\mathbf{W}(a)
\end{array}	\right]
\begin{array}{l} 
\leftarrow \mbox{idle} \\
\leftarrow \mbox{busy} 
\end{array},\\
\label{eq:trans_matrix}
\end{equation} 
where $\eta$ is the probability that the incumbent channel is idle. The first row of matrix $\mathbf{P}(a)$ corresponds to the case when the incumbent channel is idle and the second row corresponds to the case when the incumbent channel is busy. The matrix ${\mathbf{W}}(a)$ represents the state transition of the ST including both the data queue and the energy storage. As mentioned in Section~\ref{state_action_spaces}, when the incumbent channel is idle, the ST will stay idle or transmit data. Otherwise, the ST will harvest energy or backscatter data. Thus, we consider 2 cases of the channel status and derive the corresponding transition probability matrices.

\subsubsection{The incumbent channel is idle}
We first derive the transition probability matrix when the incumbent channel is idle as follows:
\begin{equation}
{\mathbf{P}}(a)	=	\left[	\begin{array}{cc}
\eta \mathbf{W}(a)	&	(1-\eta)\mathbf{W}(a)	\\
0	&	0
\end{array}	\right]
\begin{array}{l} 
\leftarrow \mbox{idle} \\
\leftarrow \mbox{busy} 
\end{array}.\\
\label{eq:trans_matrix_channel_idle}
\end{equation} 

In this case, the ST can choose to stay idle, i.e., $a =1$, or transmit data by using energy in the energy storage, i.e., $a = 2$.

\paragraph{The ST stays idle} The transition probability matrix of the ST is expressed in (\ref{W1}). Note that in this paper, the empty elements in the transition probability matrices are either zeros or zero matrices with appropriate sizes.
\begin{equation}
\label{W1}
{\mathbf{W}}(1)	\!=\! \left[	\begin{array}
{c@{\hspace{0.3em}}c@{\hspace{0.3em}}c@{\hspace{0.3em}}c@{\hspace{0.3em}}c}
\mathbf{B}_{0,0}(1)& \mathbf{B}_{0,1}(1)&  & 	\\
& \mathbf{B}_{1,1}(1)& \mathbf{B}_{1,2}(1) & 	\\
& &\ddots &	\\
& &  & \mathbf{B}_{D,D}(1)
\end{array} \right]	
\begin{array}{l}\!	\leftarrow d=0 \\\!	\leftarrow d=1	\\	\vdots \\ \!	\leftarrow d=D	\end{array}		,
\end{equation}
where each row of matrix $\mathbf{W}(1)$ corresponds to the number of packets in the data queue, i.e., the queue state. The matrix ${\mathbf{B}}_{d,d'}(1)$ represents the data queue state transition from $d$ in the current time slot to $d'$ in the next time slot. Each row of the matrix ${\mathbf{B}}_{d,d'}(1)$ corresponds to the energy level of the ST. Clearly, with action $a = 1$, the energy storage will remain the same. However, the data queue can increase by one unit if there is a packet arrival. Thus, we have
\begin{equation}
{\mathbf{B}}_{d,d}(1)	\!=\! \left[	\begin{array}
{c@{\hspace{0.3em}}c@{\hspace{0.3em}}c@{\hspace{0.3em}}c@{\hspace{0.3em}}c}
(1-\alpha) & &  & 	\\
& (1-\alpha) &  & 	\\
& &\ddots &	\\
& &  & (1-\alpha)
\end{array} \right]	
\begin{array}{l}\!	\leftarrow e=0 \\\!	\leftarrow e=1	\\	\vdots \\ \!	\leftarrow e=E	\end{array},
{\mathbf{B}}_{d,d+1}(1)	\!=\! \left[	\begin{array}
{c@{\hspace{0.3em}}c@{\hspace{0.3em}}c@{\hspace{0.3em}}c@{\hspace{0.3em}}c}
\alpha & &  & 	\\
& \alpha &  & 	\\
& &\ddots &	\\
& &  & \alpha
\end{array} \right]	
\begin{array}{l}\!	\leftarrow e=0 \\\!	\leftarrow e=1	\\	\vdots \\ \!	\leftarrow e=E	\end{array}		,
\end{equation}
where $\alpha$ is the packet arrival probability. It is important to note that when the data queue is full, incoming packets will be dropped. Thus, $\mathbf{B}_{D,D}(1) = \mathbf{I}$, where $\mathbf{I}$ is an identity matrix.
\paragraph{The ST transmits data} The transition probability matrix of the ST when $a=2$ is expressed in~(\ref{trans_matrix_action2}).
\begin{equation}
\label{trans_matrix_action2}
{\mathbf{W}}(2)	\!=\! \left[	\begin{array}
{c@{\hspace{0.3em}}c@{\hspace{0.3em}}c@{\hspace{0.3em}}c@{\hspace{0.3em}}c@{\hspace{0.3em}}c}
0	\\
&	&  \ddots	\\
&	&	& & 0		\\		\hline
\mathbf{B}_{d,d-d_\mathrm{t}}(2)&\mathbf{B}_{d,d-d_\mathrm{t}+1}(2)&\mathbf{B}_{d,d}(2)&\mathbf{B}_{d,d+1}(2)	\\
\ddots & \ddots&	\ddots&	\ddots&	\\
&	&\mathbf{B}_{D,D-d_\mathrm{t}}(2)	&\mathbf{B}_{d,D-d_\mathrm{t}+1}(2)	&\mathbf{B}_{D,D}(2)&
\end{array} \right]	
\begin{array}{l}\!	\leftarrow 	d=0	\\	\vdots	\\ \!	\leftarrow	d = (d_\mathrm{t}-1)		\\ \!	\leftarrow	d= d_\mathrm{t}	\\	\vdots\\ \!	\leftarrow	d = D		\end{array}		.
\end{equation}

Again, in each time slot, the ST will transmit $d_\mathrm{t}$ data units in the data queue, i.e., $d \geq d_\mathrm{t}$. There are four cases to derive the matrix $\mathbf{B}_{d,d'}(2)$ as follows:
\begin{equation}
\label{eq:action2_4}
{\mathbf{B}}^i(2)	\!=\! \left[	\begin{array}
{c@{\hspace{0.3em}}c@{\hspace{0.3em}}c@{\hspace{0.3em}}c@{\hspace{0.3em}}c@{\hspace{0.3em}}c}
0	\\
&	&  \ddots	\\
&	&	& & 0		\\		\hline
b_t^i&&&	\\
& \ddots &	&	&	\\
&	&	&b_t^i	& \cdots &0
\end{array} \right]	
\begin{array}{l}\!	\leftarrow 	e=0	\\	\vdots	\\ \!	\leftarrow	e = e_\mathrm{t}-1		\\ \!	\leftarrow	e= e_\mathrm{t}	\\	\vdots\\ \!	\leftarrow	e = E		\end{array}		,
\end{equation}
where $i \in \{1,2,3,4\}$ corresponds to four transition probability matrices of the data queue, i.e., $\mathbf{B}_{d,d-d_\mathrm{t}}(2)$, $\mathbf{B}_{d,d-d_\mathrm{t}+1}(2)$, $\mathbf{B}_{d,d}(2)$, and $\mathbf{B}_{d,d+1}(2)$, respectively.
\begin{itemize}
	\item The first case, i.e., $\mathbf{B}_{d,d-d_\mathrm{t}}(2)$, happens when the ST successfully transmits data to its SR with the probability $\sigma$, no packet arrives, and there is enough energy in the energy storage, i.e., $e \geq e_\mathrm{t}$. Thus, the probability for this case is $b_\mathrm{t}^1=\sigma(1-\alpha)$.	
	\item The second case, i.e., $\mathbf{B}_{d,d-d_\mathrm{t}+1}(2)$, happens when the ST successfully transmits data to its SR, a packet arrives, and there is enough energy in the energy storage, i.e., $e \geq e_\mathrm{t}$. Thus, the probability for this case is $b_\mathrm{t}^2=\sigma\alpha$.	
	\item The third case, i.e., $\mathbf{B}_{d,d}(2)$, happens when the ST unsuccessfully transmits data to its SR with the probability $(1-\sigma)$, no packet arrives, and there is enough energy in the energy storage, i.e., $e \geq e_\mathrm{t}$. Thus, the probability for this case is $b_\mathrm{t}^3=(1-\sigma)(1-\alpha$).	
	\item The fourth case, i.e., $\mathbf{B}_{d,d+1}(2)$, happens when the ST unsuccessfully transmits data to its SR, a packet arrives, and there is enough energy in the energy storage, i.e., $e \geq e_\mathrm{t}$. Thus, the probability for this case is $b_\mathrm{t}^4=(1-\sigma)\alpha$.
\end{itemize}
Note that when the data queue is full, i.e., $d=D$, there is no fourth case, the calculation of $b^i_\mathrm{t}$ for the first two cases remain unchanged, while for the third case, $b^i_\mathrm{t}=(1-\sigma)(1-\alpha) +(1-\sigma)\alpha = 1-\sigma$. There is also a special case when $d_\mathrm{t}=1$. In this case, the indexes $(d-d_\mathrm{t}+1)$ and $d$ are the same. Thus, the probability for this case is $\sigma\alpha + (1-\sigma)(1-\alpha)$.
\subsubsection{The incumbent channel is busy}
When the incumbent channel is busy, the transition probability matrix is expressed as follows:
\begin{equation}
{\mathbf{P}}(a)	=	\left[	\begin{array}{cc}
0	&	0	\\
\eta\mathbf{W}(a)	&	(1-\eta)\mathbf{W}(a)
\end{array}	\right]
\begin{array}{l} 
\leftarrow \mbox{idle} \\
\leftarrow \mbox{busy} 
\end{array}.\\
\label{eq:trans_matrix_channel_busy}
\end{equation}
In this case, the ST can choose to harvest RF energy and store it in the energy storage, i.e., $a =3$ or backscatter data in the data queue to the secondary receiver, i.e., $a =4$. 
\paragraph{The ST harvests energy}
The transition probability matrix can be expressed as follows:
\begin{equation}
{\mathbf{W}}(3)	\!=\! \left[	\begin{array}
{c@{\hspace{0.3em}}c@{\hspace{0.3em}}c@{\hspace{0.3em}}c@{\hspace{0.3em}}c}
\mathbf{B}_{0,0}(3)& \mathbf{B}_{0,1}(3)&  & 	\\
& \mathbf{B}_{1,1}(3)& \mathbf{B}_{1,2}(3) & 	\\
& &\ddots &	\\
& &  & \mathbf{B}_{D,D}(3)
\end{array} \right]	
\begin{array}{l}\!	\leftarrow d=0 \\\!	\leftarrow d=1	\\	\vdots \\ \!	\leftarrow d=D	\end{array}		.
\end{equation}
When the data queue is not full, i.e., $d<D$, as the ST can only harvest energy, there are two cases for deriving the transition matrix given as follows:
\begin{equation}
\label{eq:action3_2}
{\mathbf{B}}_{d,d+1}(3)	\!=\! \left[	\begin{array}
{c@{\hspace{1.533em}}c@{\hspace{1.533em}}c@{\hspace{1.533em}}c}
b_\mathrm{a}^\diamond &b_\mathrm{a}^\circ &   	\\
& b_\mathrm{a}^\diamond &b_\mathrm{a}^\circ &   	\\
& & \ddots& 	\\
& &  & \alpha
\end{array} \right]	
\begin{array}{l}\!	\leftarrow e=0 \\\!	\leftarrow e=1	\\	\vdots \\ \!	\leftarrow e=E	\end{array},
{\mathbf{B}}_{d,d}(3)	\!=\! \left[	\begin{array}
{c@{\hspace{1.533em}}c@{\hspace{1.533em}}c@{\hspace{1.533em}}c}
b_\mathrm{a}^\ddagger &b_\mathrm{a}^\dagger &   	\\
& b_\mathrm{a}^\ddagger &b_\mathrm{a}^\dagger &   	\\
& &\ddots&	\\
& &  &  (1-\alpha)
\end{array} \right]	
\begin{array}{l}\!	\leftarrow e=0 \\\!	\leftarrow e=1	\\	\vdots \\ \!	\leftarrow e=E	\end{array}		.	
\end{equation}
\begin{itemize}
	\item There is a packet arrival, i.e., $\mathbf{B}_{d,d+1}(3)$, with probability $\alpha$.
	\begin{itemize}
		\item The ST successfully harvests RF energy with probability $\gamma$. The probability for this case is then $b_\mathrm{a}^\circ = \alpha\gamma$.
		\item The ST unsuccessfully harvests RF energy with probability $(1-\gamma)$. The probability for this case is then $b_\mathrm{a}^\diamond = \alpha(1-\gamma)$.
	\end{itemize}
	\item There is no packet arrival, i.e., $\mathbf{B}_{d,d}(3)$, with probability denoted by $(1-\alpha)$.
	\begin{itemize}
		\item The ST successfully harvests RF energy with probability $\gamma$. The probability for this case is then $b_\mathrm{a}^\dagger = (1-\alpha)\gamma$.
		\item The ST unsuccessfully harvests RF energy with probability $(1-\gamma)$. The probability for this case is then $b_\mathrm{a}^\ddagger = (1-\alpha)(1-\gamma)$.
	\end{itemize}
\end{itemize}
When the data queue is full, i.e., $d=D$, the transition matrix $\mathbf{B}_{D,D}(3)$ is expressed as follows:
\begin{equation}
\label{eq:action3_3}
{\mathbf{B}}_{D,D}(3)	\!=\! \left[	\begin{array}
{c@{\hspace{0.3em}}c@{\hspace{0.3em}}c@{\hspace{0.3em}}c@{\hspace{0.3em}}c}
(1-\gamma) &\gamma &  & 	\\
& (1-\gamma) &\gamma &  & 	\\
& & & \ddots	\\
& &  & &1
\end{array} \right]	
\begin{array}{l}\!	\leftarrow e=0 \\\!	\leftarrow e=1	\\	\vdots \\ \!	\leftarrow e=E	\end{array}		.
\end{equation}
\paragraph{The ST backscatters data} 
The transition probability matrix of the data queue is expressed as follows:
\begin{equation}
\label{trans_matrix_action4}
{\mathbf{W}}(4)	\!=\! \left[	\begin{array}
{c@{\hspace{0.3em}}c@{\hspace{0.3em}}c@{\hspace{0.3em}}c@{\hspace{0.3em}}c@{\hspace{0.3em}}c}
0	\\
&	&  \ddots	\\
&	&	& & 0		\\		\hline
\mathbf{B}_{d,d-d_\mathrm{b}}(4)&\mathbf{B}_{d,d-d_\mathrm{b}+1}(4)&\mathbf{B}_{d,d}(4)&\mathbf{B}_{d,d+1}(4)	\\
\ddots &\ddots&\ddots&	\ddots&	\\
&	&\mathbf{B}_{D,D-d_\mathrm{b}}(4)&\mathbf{B}_{D,D-d_\mathrm{b}+1}(4)&\mathbf{B}_{D,D}(2)&
\end{array} \right]	
\begin{array}{l}\!	\leftarrow 	d=0	\\	\vdots	\\ \!	\leftarrow	d = (d_\mathrm{b}-1)		\\ \!	\leftarrow	d= d_\mathrm{b}	\\	\vdots\\ \!	\leftarrow	d = D		\end{array}		.
\end{equation}

Again, at each time slot, the ST will backscatter $d_\mathrm{b}$ packets from the data queue, i.e., $d \geq d_\mathrm{b}$. This process does not require any energy from the energy storage. Therefore, the energy state in the data queue remains the same. There are four cases to derive the matrix $\mathbf{B}_{d,d'}(4)$ as follows:

\begin{equation}
\label{eq:action4_4}
{\mathbf{B}}^i(4)	\!=\! \left[	\begin{array}
{c@{\hspace{0.3em}}c@{\hspace{0.3em}}c@{\hspace{0.3em}}c@{\hspace{0.3em}}c@{\hspace{0.3em}}c}
b_\mathrm{b}^i	\\
&	b_\mathrm{b}^i & 	\\
&	& \ddots	& & 		\\
&&&b_\mathrm{b}^i	\\
\end{array} \right]	
\begin{array}{l}\!	\leftarrow 	e=0	\\ \!	\leftarrow	e = 1		\\ 	\vdots\\ \!	\leftarrow	e = E		\end{array}		,
\end{equation}
where $i \in \{1,2,3,4\}$ corresponds to four transition probability matrices of the data queue, i.e., $\mathbf{B}_{d,d-d_\mathrm{b}}(4)$, $\mathbf{B}_{d,d-d_\mathrm{b}+1}(4)$, $\mathbf{B}_{d,d}(4)$, and $\mathbf{B}_{d,d+1}(4)$, respectively.

\begin{itemize}
	\item The first case, i.e., $\mathbf{B}_{d,d-d_\mathrm{b}}(4)$, happens when the ST successfully backscatters data to its receiver and no packet arrives. Thus, the probability for this case is $b_\mathrm{b}^1=\beta(1-\alpha)$.
	\item The second case, i.e., $\mathbf{B}_{d,d-d_\mathrm{b}+1}(4)$, happens when the ST successfully backscatters data to its receiver and a packet arrives. Thus, the probability for this case is $b_\mathrm{b}^2=\beta\alpha$.
	\item The third case, i.e., $\mathbf{B}_{d,d}(4)$, happens when the ST unsuccessfully backscatters data to its receiver and no packet arrives. Thus, the probability for this case is $b_\mathrm{b}^3=(1-\beta)(1-\alpha$).
	\item The fourth case, i.e., $\mathbf{B}_{d,d+1}(4)$, happens when the ST unsuccessfully backscatters data to its receiver and a packet arrives. Thus, the probability for this case is $b_\mathrm{b}^4=(1-\beta)\alpha$.
\end{itemize}
Note that when the data queue is full, i.e., $d=D$, there is no fourth case, the calculation of $b^i_\mathrm{b}$ for the first two cases remain unchanged, while for the third case, $b^i_\mathrm{b} = (1-\beta)(1-\alpha) +(1-\beta)\alpha = 1-\beta$. There is also a special case when $d_\mathrm{b}=1$. In this case, the indexes $(d-d_\mathrm{b}+1)$ and $d$ are the same. Thus, the probability for this case is $ \beta\alpha + (1-\beta)(1-\alpha)$.
\subsection{Optimization Formulation}
We formulate an optimization problem based on the aforementioned MDP and then obtain an optimal policy, denoted by $\Omega^*$, for the ST to maximize its throughput. The policy is a mapping from a state to an action taken by the ST. In other words, given the data queue, energy level, and incumbent channel states, the policy determines an action to maximize the average reward in terms of throughput for the ST. The optimization problem is then expressed as follows:
\begin{eqnarray} 
\label{eq:average_reward}
\max_\Omega	& &	{\mathcal{R}}(\Omega)	=	\lim_{T \rightarrow \infty} \frac{1}{T} \sum_{k=1}^{T} {\mathbb{E}} \left( {\mathcal{T}_k} (\Omega) \right),	\label{eq:cmdp_obj}
\end{eqnarray}
where ${\mathcal{R}}(\Omega)$ is the average throughput of the ST under the policy $\Omega$ and ${\mathcal{T}_k} (\Omega)$ is the immediate throughput under policy $\Omega$ at time step $k$ that is defined as follows:
\begin{equation}
{\mathcal{T}}	=	\left\{	\begin{array}{ll}
\sigma d_\mathrm{t},	&	(a=2),	\\
\beta d_\mathrm{b},	&	(a=4),	\\
0	,						&	\mbox{otherwise}	.
\end{array}	\right.
\end{equation}
In Theorem~\ref{theo:limitexists}, we show that the average throughput $\mathcal{R}(\Omega)$ is well defined and does not depend on the initial state.

\begin{theorem}
	\label{theo:limitexists}
	For every $\Omega$, the average throughput $\mathcal{R}(\Omega)$ is well defined and does not depend on the initial state.
\end{theorem}
\begin{proof}
	To prove this theorem, we first point out that the Markov chain is irreducible. It means that we need to prove that $p_{s,s'}>0, \forall s, s' \in \mathcal{S}$, i.e., the process can go from any state to any state. We will consider two cases, i.e., the incumbent channel is busy and idle, and prove that in each case, the transition probabilities will be always greater than 0. Clearly, from any state when the incumbent channel is busy (or idle), the process can move to any state when the incumbent channel is idle (or busy). Intuitively, as the probability of the incumbent channel being idle is $\eta$, from state $s=(0,d,e), \forall d \in \mathcal{D}~\mbox{and}~\forall e \in \mathcal{E}$, we can move to state $s'=(1,d,e)$ with probability $(1-\eta) > 0$. In contrast, the process will move from state $s'= (1,d,e)$ to state $s=(0,d,e)$ with probability $\eta > 0$.

	Consider the case when the incumbent channel is idle. Given state $s=(0,d,e), \forall d \in \mathcal{D}~\mbox{and}~\forall e \in \mathcal{E}$, if the ST chooses to stay idle, the system will move to state $s'=(0,d+1,e)$ with probability $\alpha$ if there is a packet arrival, and remain unchanged with probability $(1-\alpha) > 0$ if there is no packet arrival. If the ST chooses to transmit data, the system will move to the next state under the four cases with different probabilities $b^1_{\mathrm{t}} > 0 $, $b^2_{\mathrm{t}} > 0$, $b^3_{\mathrm{t}} > 0 $, and $b^4_{\mathrm{t}} > 0$ as discussed in Section~\ref{subsec:trans}. Note that the energy storage needs to have enough energy, i.e., $e \geq e_{\mathrm{t}}$, to support the transmission. Through the energy harvesting process, the system always can move to a state in which there is enough energy in the energy storage as discussed in the following.
	
	Consider the case when the incumbent channel is busy. Given state $s=(1,d,e), \forall d \in \mathcal{D}~\mbox{and}~\forall e \in \mathcal{E}$, if the ST chooses to harvest energy from the incumbent signals and there is a packet arrival, the system will move to state $s'=(1,d+1,e+1)$, i.e., successfully harvests RF energy, and $s'=(1,d+1,e)$, i.e., unsuccessfully harvests RF energy, with probabilities $b_\mathrm{a}^\circ>0$ and $b_\mathrm{a}^\diamond>0$, respectively. If there is no packet arrival, the system will move to state $s'=(1,d,e+1)$ and remain unchanged with probabilities $b_\mathrm{a}^\dagger>0$ and $b_\mathrm{a}^\ddagger>0$, respectively. For the case when the ST chooses to backscatters data, the system will move to the next state under the four cases with different probabilities $b^1_{\mathrm{b}} > 0 $, $b^2_{\mathrm{b}} > 0 $, $b^3_{\mathrm{b}} > 0 $, and $b^4_{\mathrm{b}} > 0 $ as stated in Section~\ref{subsec:trans}.
	
	Thus, the state space $\mathcal{S}$ contains only one communicating class, i.e., $p_{s,s'}>0, \forall s, s'$. In other words, the MDP with states in $\mathcal{S}$ is irreducible. As a result, the average throughput $\mathcal{R}(\Omega)$ is well defined and does not depend on the initial state~\cite{Puterman_1994_Book},~\cite{CompetitiveBook}. 
\end{proof}

Then, we obtain the optimal policy from the optimization problem by formulating and solving a linear programming (LP) problem~\cite{Puterman_1994_Book}. The LP problem is expressed as follows:
\begin{eqnarray}
\max_{\psi(s,a)}	& &	\sum_{s \in \mathcal{S}} \sum_{a \in \mathcal{A}}	\psi(s,a)	{\mathcal{T}} ( s, a )	\label{eq:obj_lp}	\\
\mbox{s.t.}			& & \sum_{a \in \mathcal{A}}	\psi(s',a)	=	 \sum_{s \in \mathcal{S}} \sum_{a \in \mathcal{A}}	\psi(s,a)	p_{s, s'} (a), \quad	\forall s' \in \mathcal{S}	\nonumber	\\
& & \sum_{s \in \mathcal{S}} \sum_{a \in \mathcal{A}}	\psi(s,a)	=1, \quad	\psi(s,a)	\geq 0	\nonumber	,
\end{eqnarray}
where $p_{s, s'} (a)$ denotes the element of matrix ${\mathbf{P}}(a)$. Let the solution of the LP problem be denoted by $\psi^*(s,a)$. Then, the policy of the ST obtained from the optimization problem is expressed as follows~\cite{Puterman_1994_Book}:
\begin{equation}
\Omega^*(s,a)	=	\frac{ \psi^*(s,a) }{ \sum_{ a' \in \mathcal{A} } \psi^*(s,a') }, \quad \forall s \in \mathcal{S} .
\label{eq:optimal_policy}
\end{equation}
\subsection{Performance Metrics}
Given that the optimization problem is feasible, we can obtain the optimal policy for the ST. The following performance measures then can be derived.

\paragraph{Average number of packets in the data queue is obtained from}
\begin{equation}
\overline{d}		=		\sum_{	a \in \mathcal{A}	}	\sum_{	c \in \mathcal{C}	} \sum_{d=0}^D	\sum_{e=0}^E	d	\psi^*( (c,d,e) ,a)	.
\end{equation}

\paragraph{Average throughput is obtained from}
\begin{equation}
\tau	= \sum_{d=d_\mathrm{t}}^D	\sum_{e=e_\mathrm{t}}^E	\sigma d_\mathrm{t}	\psi^*( (c,d,e) ,2)		+ \sum_{d=d_\mathrm{b}}^D	\sum_{e=0}^E	\beta d_\mathrm{b}	\psi^*( (c,d,e) ,4).
\end{equation}

\paragraph{Average delay can be obtained using Little's law as follows}
\begin{equation}
\overline{\kappa}	=	\frac{\overline{d}}	{ \tau }, 
\end{equation}
where $\tau$ is the effective arrival rate which is the same as the throughput.	

\section{Proposed Low-Complexity Online Reinforcement Learning Algorithm}
\label{sec:LA}

The aforementioned MDP introduced in Section~\ref{sec:markov} requires environment parameters, e.g., channel idle and packet arrival probabilities, to construct transition probability matrices. Nevertheless, these environment parameters may not be available for formulating the optimization problem in practice. We thus propose a low-complexity reinforcement learning algorithm to obtain the optimal policy for the ST in an online fashion without requiring the environment parameters in advance. In particular, we implement the online reinforcement learning algorithm on the ST that directly guides the controller to take actions as shown in Fig.~\ref{Fig.learning}. Given the current state and its policy, the learning algorithm will make an optimal decision and send to the controller to perform the action. After that, the learning algorithm observes the results and updates its current policy. In this way, the learning algorithm can improve its policy, and we will show that our proposed learning algorithm can converge to the optimal policy. 
\begin{figure}[tbh]
	\centering
	\includegraphics[scale=0.3]{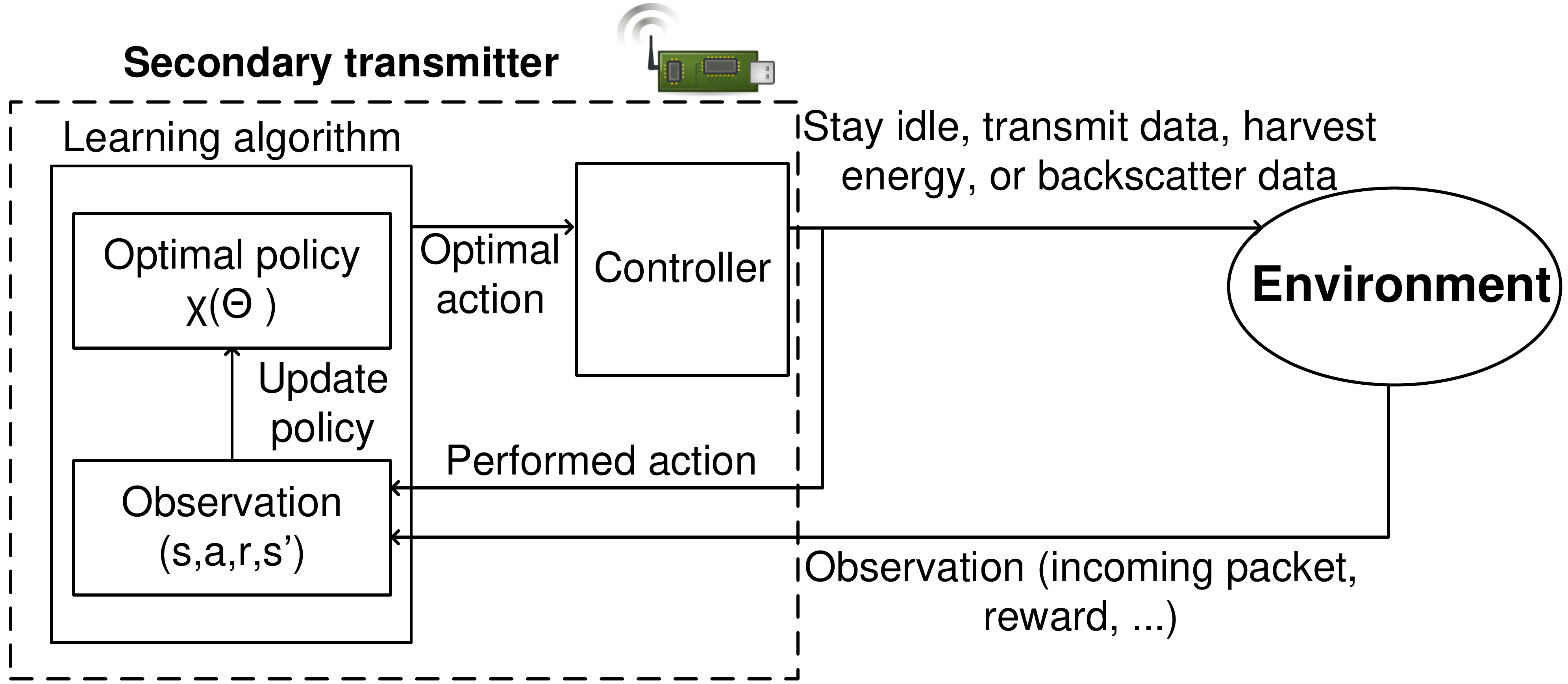}
	\caption{The learning model.}
	\label{Fig.learning}
\end{figure}

\subsection{Parameterization for the MDP}
We denote $\Theta = \{\theta_{ s,a } \in \mathbb{R}\}$ as a parameter vector of the ST at state $s$ with the current action $a$. We consider a randomized parameterized policy~\cite{Marbach2001, Baxter2001} to find decisions for the ST. Under the randomized parameterized policy, when the ST is at state $s$, it will choose action $a$ with the probability $\chi_{\Theta}(s,a)$ which is normalized as follows:
\begin{equation}
\chi_{\Theta}(s,a) = \frac{\exp\big(\theta_{ s,a }\big)} {\sum_{a' \in \mathcal{A}}\exp\big(\theta_{ s,a'  }\big)},
\label{eq:randomized_parameterized_function}
\end{equation}
where $\Theta = \left[	\begin{array}{ccc}	\cdots	&	\theta_{s, a}	&	\cdots	\end{array}	\right]^\top$ is used to support the ST to make decisions given its current state. By using the results obtained from interacting with the environment, this parameter vector will be updated iteratively. In addition, the parameterized randomized policy $\chi_{\Theta}(s,a)$ must not be negative and satisfies the following constraint:
\begin{equation}
\label{eq:condition_1}
\sum_{a \in \mathcal{A}} \chi_{\Theta}(s,a) = 1	. 
\end{equation} 

Based on the parameterized randomized policy, the immediate throughput function of the secondary user is then parameterized as follows:
\begin{equation}
\label{eq:average_throguhput_defination}
{\mathcal{T}}_{\Theta} ( s )	= \sum_{a \in \mathcal{A}} \chi_{\Theta}(s, a){\mathcal{T}}( s, a ) ,
\end{equation}	
where $\mathcal{T}( s, a )$ is the immediate throughput when the ST chooses action $a$ given state $s$. Similarly, the parameterized transition probability function given the randomized parameterized policy $\chi_{\Theta}(s,a)$ can also be derived as follows:
\begin{equation}
p_{\Theta}({s,s'	})= \sum_{a \in \mathcal{A}} \chi_{\Theta}(s, a)	p_{s,s'}(  a ), \quad \forall s, s' \in \mathcal{S}, 
\end{equation}
where $p_{s,s'}( a )$ is the transition probability from state $s$ to state $s'$ when action $a$ is taken.

After that, the average throughput of the ST can be parameterized as follows:
\begin{equation}
\label{eq:throughput}
\xi(\Theta) = \lim_{t\rightarrow \infty} \frac {1}{t} \mathbb{E}_{\Theta} \Big[ \sum_{k=0}^{t} {\mathcal{T}}_{\Theta} ( s_k )\Big] ,
\end{equation}
where $s_k$ is the state of the ST at time step $k$. $\mathbb{E}_{\Theta}[ \cdot ]$ is the expectation of the throughput. Then, we derive the following proposition~\cite{Marbach2001}:
\begin{proposition}
	\label{recurrent_state}
	The Markov chain corresponding to every $\Theta$ is aperiodic. Additionally, there exists a state $s^{\dagger}$ that is recurrent for the Markov chain.
\end{proposition}

\begin{proof}
	Given state $s$ and parameter vector $\Theta$, the system will remain unchanged with the one-step probability as follows:
	\begin{equation}
		p_{\Theta}({s,s	})= \sum_{a \in \mathcal{A}} \chi_{\Theta}(s, a)	p_{s,s}(a)= \sum_{a \in \mathcal{A}}\frac{\exp(\theta_{ s,a})}{\sum_{ a' \in \mathcal{A} }\exp(\theta_{ s,a'  })}p_{s,s}(a), \forall s \in \mathcal{S}.
	\end{equation}
	Clearly, as mentioned in the proof of Theorem~\ref{theo:limitexists} and Section~\ref{subsec:trans}, $p_{\Theta}(s,s) > 0, \forall s \in \mathcal{S}$. As a result, the period of state $s \in \mathcal{S}$, i.e., the greatest common divisor of all $n$-step transitions of state $s$, is equal to $1$, thereby state $s$ is aperiodic. As every state $s \in \mathcal{S}$ is aperiodic, the Markov chain corresponding to $\Theta$ is also aperiodic.
	
	Similar to the the proof of Theorem~\ref{theo:limitexists}, we can show that the Markov chain corresponding to $\Theta$ is irreducible. Therefore, there always exists a recurrent state $s^\dagger$ for the Markov chain corresponds to every $\Theta$.
\end{proof}

Proposition~\ref{recurrent_state} implies that the average throughput $\xi(\Theta)$ is well defined for every $\Theta$. More importantly, $\xi(\Theta)$ does not depend on the initial state $s_0$. Additionally, we have the following balance equations:
\begin{equation}
\label{eq: balance equation}
\sum_{s \in \mathcal{S}} \pi_{\Theta}({s})=1 \quad \mbox{and} \quad \sum_{  s \in \mathcal{S} } \pi_{\Theta}({s}) p_{\Theta}({s,s'}) = \pi_{\Theta}({s'}), \quad \forall s' \in \mathcal{S},
\end{equation}
where $\pi_{\Theta}({s})$ is the steady-state probability of state $s$ under the parameter vector $\Theta$. The balance equations (\ref{eq: balance equation}) have a solution denoted by a vector $\Pi_{\Theta} = \left[	\begin{array}{ccc}	\cdots	&	\pi_{\Theta}({s})	&	\cdots	\end{array}	\right]^\top$~\cite{Marbach2001}. From (\ref{eq:throughput}) and (\ref{eq: balance equation}), we can derive the parameterized average throughput as follows:
\begin{equation}
\label{eq:average throughput}
\xi(\Theta) = \sum_{s	\in		\mathcal{S}} \pi_{\Theta}({s}) {\mathcal{T}}_{\Theta} ( s )	.
\end{equation}
The objective of the optimal policy is to find the optimal value of $\Theta$ to maximize the average throughput $\xi(\Theta)$ of the ST.

\subsection{Policy Gradient Method}
To obtain the gradient of the average throughput $\xi(\Theta)$, we define the differential throughput $d(s,\Theta)$ at state $s$. Note that this differential throughput is used to show the relation between the immediate throughput and the average throughput of the ST at state $s$ instead of the recurrent state $s^\dagger$. Then, the differential throughput $d(s,\Theta)$ is expressed as follows:
\begin{equation}
d(s, \Theta) = \mathbb{E}_{\Theta} \left[ \sum_{k=0}^{T-1} \left( 	{\mathcal{T}}_{\Theta} ( s_k ) - \xi(\Theta)	\right) | s_{0} = s \right],
\end{equation}
where $T=\min\{k>0 | s_k = s^\dagger\}$ is the first next time that the system revisits the recurrent state $s^\dagger$. Under Proposition~\ref{recurrent_state}, the differential throughput $d(s,\Theta)$ is a unique solution of the following Bellman equation: 
\begin{equation}
d(s, \Theta)={\mathcal{T}}_{\Theta}( s )-\xi(\Theta)+\sum_{s' \in \mathcal{S}} p_{\Theta}({s,s'}) d(s', \Theta), \quad \forall s \in \mathcal{S}.
\end{equation}

We then make the following proposition:

\begin{proposition}
	\label{derivatives}
	For any two states $s, s' \in \mathcal{S}$, the immediate throughput function ${\mathcal{T}}_{\Theta} (s)$ and the transition probability function $p_{\Theta}({s,s'})$ satisfy the following conditions: (1) twice differentiable and (2) the first and second derivatives with respect to $\theta$ are bounded.
\end{proposition}
The proof of Proposition~\ref{derivatives} is provided in Appendix~\ref{appendix:pro_derivatives}. In particular, Proposition~\ref{derivatives} ensures that the immediate reward and the transition probability functions depend ``smoothly'' on $\theta$. With the differential throughput $d(s,\Theta)$, the gradient of the average throughput $\xi(\Theta)$ can be easily derived as stated in Theorem~\ref{prop_policy_gradient}.

\begin{theorem}
	\label{prop_policy_gradient}
	Under Proposition~\ref{recurrent_state} and Proposition~\ref{derivatives}, we have
	\begin{equation}
	\nabla \xi(\Theta) = \sum_{s \in \mathcal{S}} \pi_{\Theta}(s) \Big(\nabla {\mathcal{T}}_{\Theta} ( s ) + \sum_{s' \in \mathcal{S}} \nabla p_{\Theta}(s,s') d(s',\Theta)  \Big) .
	\end{equation}
\end{theorem}
The proof of Theorem~\ref{prop_policy_gradient} is provided in Appendix~\ref{appendix:prop_policy_gradient}.

\subsection{Idealized Gradient Algorithm}
As stated in~\cite{Bertsekas1995}, the idealized gradient algorithm is formulated based on results obtained in Theorem~\ref{prop_policy_gradient} as follows:
\begin{equation}
\label{idealized_algorithm_theta}
\Theta_{k+1} = \Theta_{k} + \rho_{k} \nabla \xi(\Theta_{k}),
\end{equation}
where $\rho_{k}$ is a step size. To guarantee the convergence of the algorithm, the step size $\rho_{k}$ must be nonnegative, deterministic, and satisfies the following constraints:
\begin{equation}
\label{eq:step_size}
\sum_{k=1}^{\infty}\rho_{k} = \infty, \mbox{ and } \sum_{k=1}^{\infty} ( \rho_{k} )^{2} < \infty	.
\end{equation}
Specifically, the step size has to approach to zero when the time step approaches to infinity. For the policy gradient method, the algorithm will begin with an initial parameter vector $\Theta_{0} \in \mathfrak{R}^{|\mathcal{S}|}$, and the parameter vector $\Theta$ will be adjusted at each time step by using ~(\ref{idealized_algorithm_theta}). Under Proposition~\ref{derivatives} in~\cite{Bertsekas1995}, it is proved that $\lim_{k \rightarrow \infty} \nabla \xi(\Theta_{k}) = 0$, and thus $\xi(\Theta_{k})$ converges.

\subsection{Online Reinforcement Learning Algorithm}

By calculating the gradient of the function $\xi(\Theta_{k})$ with respect to $\Theta$ at each time step $k$, the average throughput $\xi(\Theta_{k})$ can be maximized based on the idealized gradient algorithm. Nevertheless, the gradient of the average throughput $\xi(\Theta_{k})$ may not be exactly calculated if the size of the state space $\mathcal{S}$ is very large. Therefore, we propose the online reinforcement learning algorithm which can estimate the gradient $\xi(\Theta_{k})$ and update the parameter vector $\Theta$ at each time step in an online fashion.

From~(\ref{eq:condition_1}), as $\sum_{a \in \mathcal{A}} \chi_{\Theta}(s,a) = 1$, we can obtain that $\sum_{a \in \mathcal{A}} \nabla \chi_{\Theta}(s,a) = 0$. Therefore, from~(\ref{eq:average_throguhput_defination}), we have 
\begin{equation}
\begin{aligned}
\nabla {\mathcal{T}}_{\Theta} (s)  = & \sum_{a \in \mathcal{A}} \nabla \chi_{\Theta}(s, a) {\mathcal{T}} ( s, a)
= 	\sum_{a \in \mathcal{A}} \nabla \chi_{\Theta}(s, a) ({\mathcal{T}} ( s, a )  - \xi(\Theta)).
\end{aligned}
\end{equation}

In addition, for all $s \in \mathcal{S}$, we have 
\begin{equation}
\sum_{s' \in \mathcal{S}}\nabla p_{\Theta}({s,s'}) d(a',\Theta) = \sum_{s' \in \mathcal{S}} \sum_{a \in \mathcal{A}} \nabla \chi_{\Theta}(s, a) p_{a}({s,s'}) d(s',\Theta).
\end{equation}

Thus, under Theorem~\ref{prop_policy_gradient}, the gradient of $\xi(\Theta)$ can be expressed as follows:
\begin{equation}
\begin{aligned}
\nabla \xi(\Theta)  = &	\sum_{s \in \mathcal{S}} \pi_{\Theta}(s) \Big(\nabla {\mathcal{T}}_{\Theta} ( s ) + \sum_{s' \in \mathcal{S}}\nabla p_{\Theta}({s,s'}) d(s',\Theta) \Big) \\
= &	\sum_{s \in \mathcal{S}} \pi_{\Theta}(s)\Big(\sum_{a \in \mathcal{A}} \nabla \chi_{\Theta}(s, a) \big({\mathcal{T}} ( s, a ) - \xi(\Theta)\big)		
+ \sum_{s' \in \mathcal{S}}\sum_{a \in \mathcal{A}} \nabla \chi_{\Theta}(s, a) p_{a}({s,s'}) d(s',\Theta) \Big) \\
= &	\sum_{s \in \mathcal{S}} \pi_{\Theta}(s) \sum_{a \in \mathcal{A}} \nabla \chi_{\Theta}(s, a)  \Big(\big({\mathcal{T}} ( s, a ) - \xi(\Theta)\big) 
+ \sum_{s' \in \mathcal{S}} p_{a}(s, s') d(s', \Theta) \Big) \\ 
= &	\sum_{s \in \mathcal{S}} \sum_{a \in \mathcal{A}} \pi_{\Theta}(s) \nabla \chi_{\Theta}(s, a) q_{\Theta}(s,a),
\end{aligned}
\end{equation}
where 
\begin{equation}
\begin{aligned}
q_{\Theta}(s,a)&= \Big({\mathcal{T}} ( s, a ) - \xi(\Theta)\Big) + \sum_{s' \in \mathcal{S}}p_{a}({s,s'}) d(s',\Theta) = \mathbb{E}_{\Theta} \Bigg[\sum_{k=0}^{T-1}\big( {\mathcal{T}} ( s_k, a_k) - \xi(\Theta) \big) | s_{0} = s, a_{0}=a \Bigg].
\end{aligned}
\end{equation}

Here $T=\min\{k>0 | s_{k}=s^\dagger\}$ is the first future time that the learning algorithm visits the recurrent state $s^\dagger$. In addition, $q_{\Theta}(s,a)$ can be expressed as the differential throughput if the ST chooses action $a$ at state $s$ based on policy $\chi_{\Theta}$. Then, we introduce Algorithm~\ref{algorithm0} that updates the parameter vector $\Theta$ at each time it visits the recurrent state $s^\dagger$ as follows. 
\renewcommand{\baselinestretch}{1}
\begin{algorithm}
	\caption{Algorithm to update parameter vector $\Theta$ at recurrent state $s^\dagger$}
	\label{algorithm0}
	\begin{algorithmic}[1]
		\State \textbf{Inputs:} $\nu$, $\rho_m$, and $\Theta_{0}$.
		\State \textbf{Initialize:} initiate parameter vector $\Theta_{0}$ and randomly select a policy for the ST.
		\For{\textit{k=1 to T}}
		\State{Update current state $s$}
		\If{$s_k \equiv s^\dagger$}
		\begin{equation}
		\label{theta_al1}
		\Theta_{m+1} = \Theta_{m} + \rho_{m}F_{m}(\Theta_{m},\widetilde{\xi}_{m}),
		\end{equation}
		\begin{equation}
		\label{xi_al1}
		\widetilde{\xi}_{m+1} = \widetilde{\xi}_{m} + \nu \rho_{m}\sum_{k'=k_{m}}^{k_{m+1}-1}\Big({\mathcal{T}}(s_{k'}, a_{k'}) - \widetilde{\xi}_{m}\Big),
		\end{equation}
		\quad \quad where
		\begin{equation}
		F_{m}(\Theta_{m},\widetilde{\xi}_{m}) = \sum_{k'=k_{m}}^{k_{m+1}-1} \widetilde{q}_{\Theta_{m}}(s_{k'},a_{k'}) \frac{\nabla \chi_{\Theta_{m}}(s_{k'},a_{k'})}{\chi_{\Theta_{m}}(s_{k'},a_{k'})},
		\end{equation}
		\begin{equation}
		\widetilde{q}_{\Theta_{m}}(s_{k'},a_{k'}) = \sum_{k=k'}^{k_{m+1}-1}\Big({\mathcal{T}} (s_{k}, a_{k}) - \widetilde{\xi}_{m}\Big).
		\end{equation}
		\State{$m=m+1$}
		\EndIf
		\State{Update $\rho_m$}
		\EndFor
		\State {\textbf{Outputs:}} The optimal value of $\Theta$
	\end{algorithmic}
\end{algorithm}
\renewcommand{\baselinestretch}{1.533}
In Algorithm~\ref{algorithm0}, the step size $\rho_{m}$ satisfies (\ref{eq:step_size}) and $\nu$ is a positive constant. $F_{m}(\Theta_{m},\widetilde{\xi}_{m})$ is the estimated gradient of the average throughput calculated by the cumulative sum of the total estimated gradient of the average throughput between the $m$-th and $(m+1)$-th visits of the algorithm to the recurrent state $s^\dagger$. Additionally, the gradient of the randomized parameterized policy function in (\ref{eq:randomized_parameterized_function}) is derived as $\nabla \chi_{\Theta_{m}}(s_{k'},a_{k'})$. Through Algorithm~\ref{algorithm0}, the parameter vector $\Theta$ and the estimated average throughput $\widetilde{\xi}$ are adjusted at each time step. Then, the convergence result of Algorithm~\ref{algorithm0} is derived as in Theorem~\ref{prop2}.
\begin{theorem}
	\label{prop2}
	Let $(\Theta_{0}, \Theta_{1}, \ldots, \Theta_{\infty})$ be a sequence of the parameter vectors generated by Algorithm~\ref{algorithm0}. Then, $\xi({\Theta_{m}})$ converges and 
	\begin{equation}
	\lim_{m\rightarrow \infty} \nabla \xi(\Theta_{m}) = 0,
	\end{equation}
	with probability one. 
\end{theorem}
The proof of Theorem~\ref{prop2} is provided in Appendix~\ref{appendix:prop2}.

With Algorithm~\ref{algorithm0}, we need to store all values of $\frac{\nabla \chi_{\Theta_{m}}(s_{k},a_{k})}{\chi_{\Theta_{m}}(s_{k},a_{k})}$ and  $\widetilde{q}_{\Theta_{m}}(s_{k},a_{k})$ between the $m$-th and $(m+1)$-th visits in order to update the values of the parameter vector $\Theta$. This may lead to slow processing especially when the size of the state space $\mathcal{S}$ is large. To deal with this issue, Algorithm~\ref{algorithm0} is modified to be able to update parameter vectors at every time slot with simple calculations. First, $F_{m}(\Theta_{m},\widetilde{\xi}_{m})$ is reformulated as follows:
\begin{equation}
\begin{aligned}
F_{m}(\Theta_{m},\widetilde{\xi}_{m})&= \sum_{k'=k_{m}}^{k_{m+1}-1} \widetilde{q}_{\Theta_{m}}(s_{k'},a_{k'}) \frac{\nabla \chi_{\Theta_{m}}(s_{k'},a_{k'})}{\chi_{\Theta_{m}}(s_{k'},a_{k'})} \\
&= \sum_{k'=k_{m}}^{k_{m+1}-1}\frac{\nabla \chi_{\Theta_{m}}(s_{k'},a_{k'})}{\chi_{\Theta_{m}}(s_{k'},a_{k'})} \sum_{k=k'}^{k_{m+1}-1}\big( {\mathcal{T}} (s_{k}, a_{k}) - \widetilde{\xi}_{m}\big)= \sum_{k'=k_{m}}^{k_{m+1}-1} \big( {\mathcal{T}} (s_{k}, a_{k}) - \widetilde{\xi}_{m}\big) z_{k+1}, 
\end{aligned}
\end{equation}
where
\begin{equation}
z_{k+1} = \left\{ 
\begin{array}{ll}
\frac{\nabla \chi_{\Theta_{m}}(s_{k},a_{k})}{\chi_{\Theta_{m}}(s_{k},a_{k})}, & \text{if} \phantom{1} k = k_{m},\\
z_{k}+\frac{\nabla \chi_{\Theta_{m}}(s_{k},a_{k})}{\chi_{\Theta_{m}}(s_{k},a_{k})}, & k=k_{m}+1,\ldots,k_{m+1}-1. \\
\end{array}
\right.
\label{eq:zupdate_time}
\end{equation}

Then, the algorithm now can be expressed as in Algorithm~\ref{algorithm1}, where $\nu$ is a positive constant and $\rho_{k}$ is the step size of the algorithm. 
\renewcommand{\baselinestretch}{1}
\begin{algorithm}
	\caption{Low-complexity algorithm to update $\Theta$ at every time step}
	\label{algorithm1}
	\begin{algorithmic}[1]
		\State \textbf{Inputs:} $\nu$, $\rho_k$, and $\Theta_{0}$.
		\State \textbf{Initialize:} initiate parameter vector $\Theta_{0}$ and randomly select a initial policy for the ST.
		\For{\textit{k=1 to T}}
		\State{Update current state $s_k$}
		\State{}
		\begin{equation}
		z_{k+1} = \left\{ 
		\begin{array}{ll}
		\frac{\nabla \chi_{\Theta_{k}}(s_{k},a_{k})}{\chi_{\Theta_{k}}(s_{k},a_{k})}, & \text{if} \phantom{1} s_{k} = s^\dagger,\\
		z_{k}+\frac{\nabla \chi_{\Theta_{k}}(s_{k},a_{k})}{\chi_{\Theta_{k}}(s_{k},a_{k})}, & \text{otherwise,} \\
		\end{array}
		\right.
		\label{eq:zupdate_state}
		\end{equation}
		\begin{equation}
		\Theta_{k+1} = \Theta_{k} + \rho_{k}({\mathcal{T}}( s_k, a_k ) -\widetilde{\xi}_{k})z_{k+1}, 
		\label{eq:updatetheta}
		\end{equation}
		\begin{equation}
		\widetilde{\xi}_{k+1} = \widetilde{\xi}_{k} + \nu\rho_{k}({\mathcal{T}}( s_k, a_k ) - \widetilde{\xi}_{k}).
		\label{eq:updatexi}
		\end{equation}
		\State{Update $\rho_k$}
		\EndFor
		\State {\textbf{Outputs:}} The optimal value of $\Theta$
		
	\end{algorithmic}
\end{algorithm}
\renewcommand{\baselinestretch}{1.533}
\subsection{Complexity Analysis}
Our proposed learning algorithm, i.e., Algorithm~\ref{algorithm1}, is very computationally efficient in terms of time and storage. This is due to the fact that the state and action variables in the algorithm can be updated iteratively in an online fashion without storing and using any information from history.

First, for the storage complexity, in Algorithm~\ref{algorithm1}, the ST just needs to update these parameter vectors, i.e., $\Theta_{k}$, $z_k$, and $\widetilde{\xi}_{k}$, where $\Theta_{k}$ is the parameter vector of the ST that we need to optimize, $z_k$ is an auxiliary variable used to compute the value of $\frac{\nabla \chi_{\Theta_{k}}(s_{k},a_{k})}{\chi_{\Theta_{k}}(s_{k},a_{k})}$ before it is used to update the value of $\Theta_{k}$, and $\widetilde{\xi}_{k}$ is the estimated value of the average throughput. Since Algorithm~\ref{algorithm1} works in an online fashion, these variables will be updated at each step, and we do not need to store all other values in the past. Thus, we can avoid the curse-of-storage which happens in the algorithms using values in history to make decisions.

Second, for the time complexity, the ST needs to do four steps in each time slot. Specifically, in the first step, based on the value of $\Theta$ in the previous step, the ST calculates the value of $\chi_{\Theta}(s,a)$ by using (\ref{eq:randomized_parameterized_function}) and decides which action to take. In the second step, the ST computes the value of $\frac{\nabla \chi_{\Theta_{k}}(s_{k},a_{k})}{\chi_{\Theta_{k}}(s_{k},a_{k})}$ and updates the value of $z$ as in (\ref{eq:zupdate_state}). In the last two steps, the ST updates the values of $\Theta$ and $\widetilde{\xi}$ as in (\ref{eq:updatetheta}) and (\ref{eq:updatexi}), respectively. Additionally, we note a very interesting point here for the calculation of (\ref{eq:zupdate_state}). Because of the special structure of $\chi_{\Theta}(s,a)$ as shown in (\ref{eq:zupdate_state}), instead of calculating the value of $\frac{\nabla \chi_{\Theta_{k}}(s_{k},a_{k})}{\chi_{\Theta_{k}}(s_{k},a_{k})}$ directly, we can transform it into an equivalent form by $1-\chi_{\Theta}(s,a)$ through few steps of mathematical manipulation. This can reduce the computation time considerably. 
 

\section{Performance Evaluation}
\label{sec:PE}

\subsection{Parameter Setting}
We perform intensive simulations to evaluate the performance of the proposed solutions under different parameter settings. In particular, when the channel is busy, we assume that if the ST harvests energy, it can successfully harvest one unit of energy with probability $\gamma = 0.9$. Otherwise, if the ST performs backscattering to transmit data, it can successfully transmit one unit of data with probability $\beta = 0.9$. When the channel is idle and if the ST wants to transmit data actively, the ST requires one unit of energy to transmit two units of data. The successful data transmission for the harvest-then-transmit mode is also assumed to be $\sigma = 0.9$. The probabilities $\gamma, \beta, \mbox{and } \sigma$ can be derived through experiments. Note that our proposed online learning algorithm does not require these information in advance. It can learn the dynamics of the environment to obtain the optimal policy for the ST. The maximum data size and the energy storage capacity are set to be 10 units. Unless otherwise stated, the idle channel probability and the packet arrival probability are 0.5. For the learning algorithm, i.e., Algorithm~\ref{algorithm1}, we use the following parameters for the performance evaluation. At the beginning, the ST will start with a randomized policy, i.e., stay idle or transmit data if the incumbent channel is idle, and harvest energy or backscatter data otherwise. We set the initial value of $\rho=10^{-5}$ and it will be updated after every 18,000 iterations as follows: $\rho_{k+1} = 0.9 \rho_{k}$~\cite{Marbach2001}. We also set $\nu = 0.01$. To evaluate the proposed solutions, we compare their performance with three other schemes.
\begin{itemize}
	\item \textit{Harvest-then-transmit (HTT):} this scheme lets the ST harvest energy when the channel is busy and transmit data when the channel becomes idle~\cite{park2013}.
	\item \textit{Backscatter communication:} with this scheme, the ST will only backscatter to transmit data when the incumbent channel is busy~\cite{LiuAmbient2013,Parks2014Turbocharging}.
	\item \textit{Random policy:} when the incumbent channel is idle, the ST will decide to stay idle or transmit data with the same probability, i.e., $0.5$. Similarly, when the incumbent channel is busy, the ST will either harvest energy or backscatter with the same probability of $0.5$.
\end{itemize}

\subsection{Numerical Results}

\subsubsection{Optimal Policy Obtained by MDP Optimization with Complete Information}
We first discuss the optimal policy obtained by the linear programming technique presented in Section~\ref{sec:markov}, and study how environment parameters impact the decisions of the ST. Figs.~\ref{fig:policy_02} and~\ref{fig:policy_08} show the optimal policy of the ST when the idle channel probability is low ($\eta = 0.2$) and high ($\eta = 0.8$), respectively. As shown in Fig.~\ref{fig:policy_02}, when the idle channel probability is $0.2$, i.e., the channel is often busy, the ST will only harvest energy when the energy state is low and the data state is high. However, when the idle channel probability is $0.8$, i.e., the channel is often idle, the ST will backscatter only if the energy storage is full and the data queue is not empty as shown in Fig.~\ref{fig:policy_08}. The reason is that when the ST actively transmits data, the ST can transmit two packets, and thus this policy is to reserve energy for the ST to transmit data when the channel is idle. 
\begin{figure*}[tbh]
	\centering
	\begin{subfigure}[b]{0.4\textwidth}
		\centering
		\includegraphics[scale=0.35]{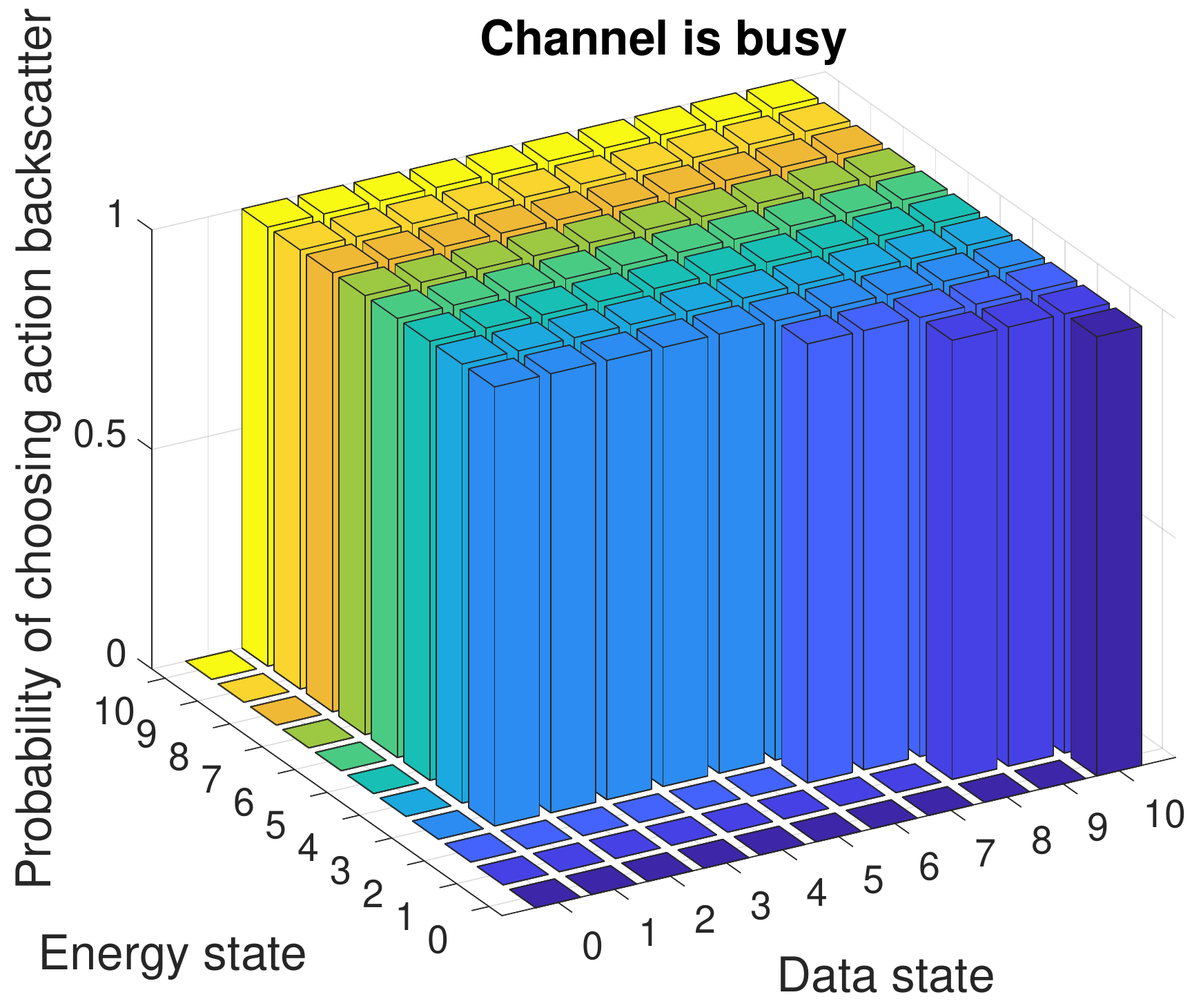}
		\caption{Channel is busy}
	\end{subfigure}%
	~ 
	\begin{subfigure}[b]{0.4\textwidth}
		\centering
		\includegraphics[scale=0.35]{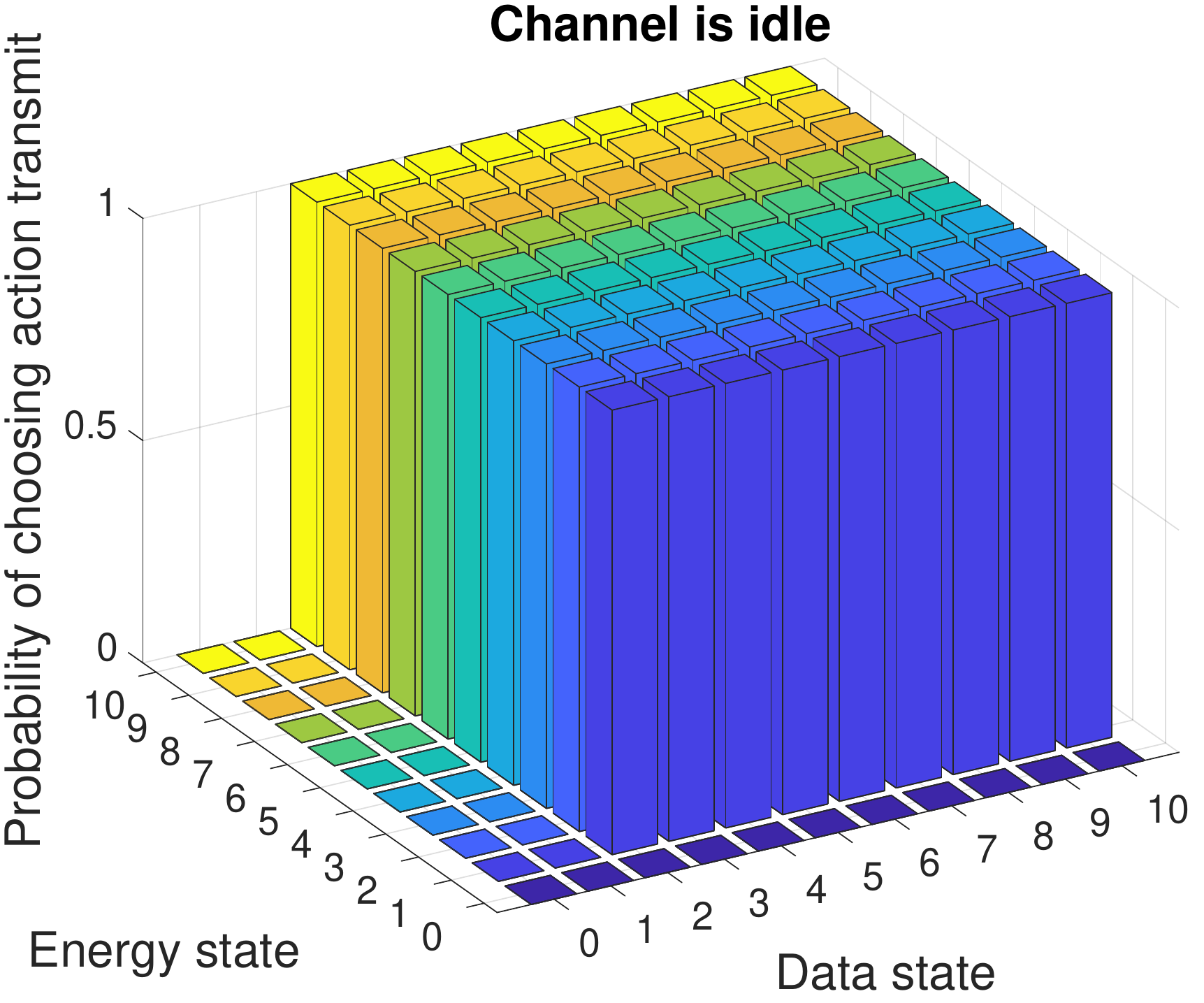}
		\caption{Channel is idle}
	\end{subfigure}
	\caption{Optimal policy of the ST when the channel idle probability $\eta$ = 0.2.} 
	\label{fig:policy_02}
\end{figure*}
\begin{figure*}[tbh]
	\centering
	\begin{subfigure}[b]{0.4\textwidth}
		\centering
		\includegraphics[scale=0.35]{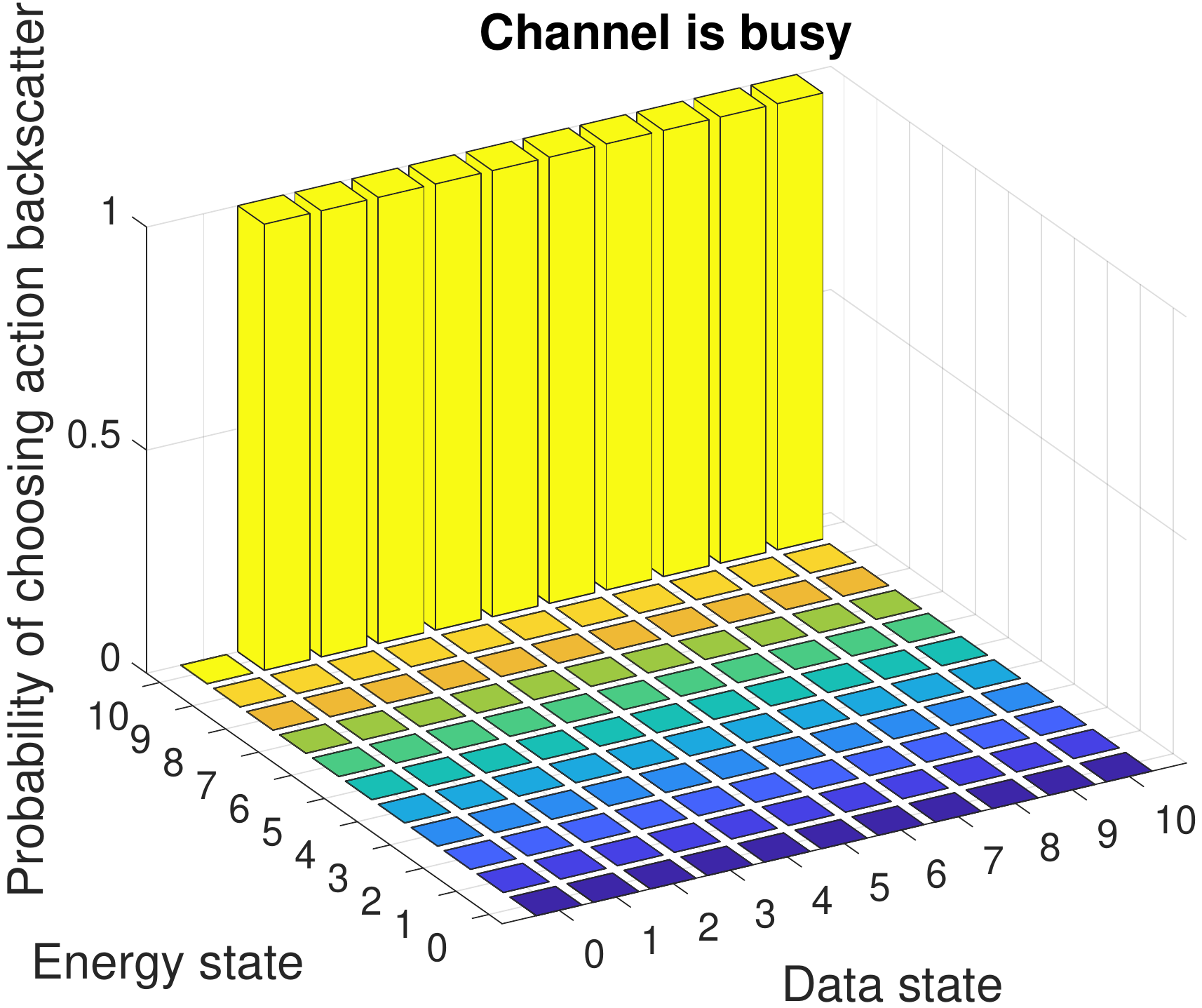}
		\caption{Channel is busy}
	\end{subfigure}%
	~ 
	\begin{subfigure}[b]{0.4\textwidth}
		\centering
		\includegraphics[scale=0.35]{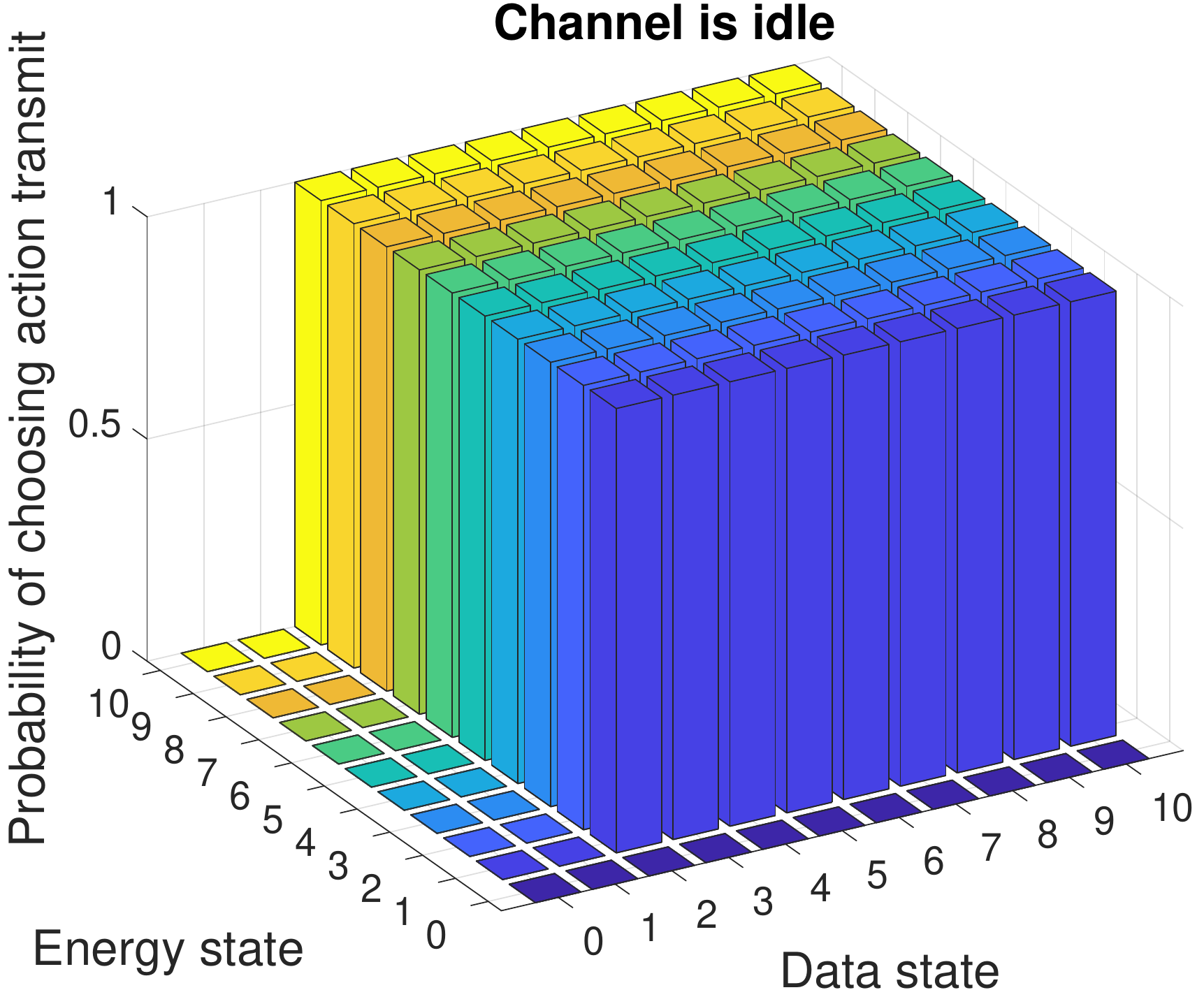}
		\caption{Channel is idle}
	\end{subfigure}
	\caption{Optimal policy of the ST when the channel idle probability $\eta$ = 0.8} 
	\label{fig:policy_08}
\end{figure*}

\subsubsection{Convergence of the Learning Algorithm}
Next, we show the learning process and the convergence of the proposed learning algorithm, i.e., Algorithm~\ref{algorithm1}, presented in Section~\ref{sec:LA}. First, we obverse the learning process of Algorithm~\ref{algorithm1} in the first 10,000 iterations. As shown in Fig.~\ref{fig:convergence}(a). In the first 4,000 iterations, the ST is still in the learning process to find the optimal values for the parameter vector $\Theta$, and thus its performance is fluctuated. However, after 4,000 iterations, the learning algorithm begins stabilizing, and thus its average throughput starts to increase. The average throughput of the ST achieves $0.68$ after $10^5$ iterations, and converges to $0.69$ after $5 \times 10^5$ iterations as shown in Fig.~\ref{fig:convergence}(b). The convergence result in Fig.~\ref{fig:convergence} also verifies our proof of convergence for the learning algorithm presented in Appendix~\ref{appendix:prop2}.

\begin{figure*}[tbh]
	\centering
	\begin{subfigure}[b]{0.45\textwidth}
		\centering
		\includegraphics[scale=0.37]{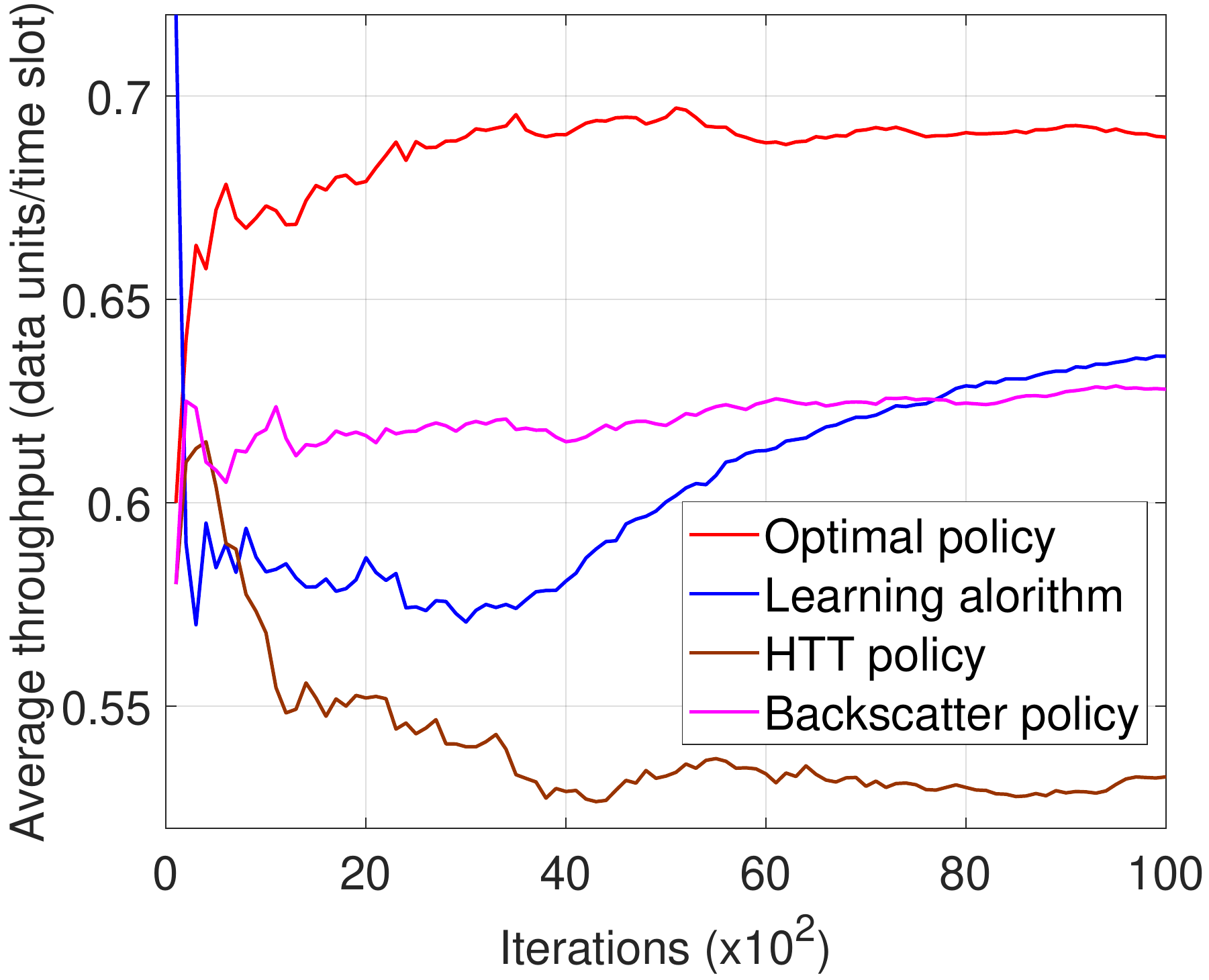}
		\caption{}
	\end{subfigure}%
	~ 
	\begin{subfigure}[b]{0.45\textwidth}
		\centering
		\includegraphics[scale=0.37]{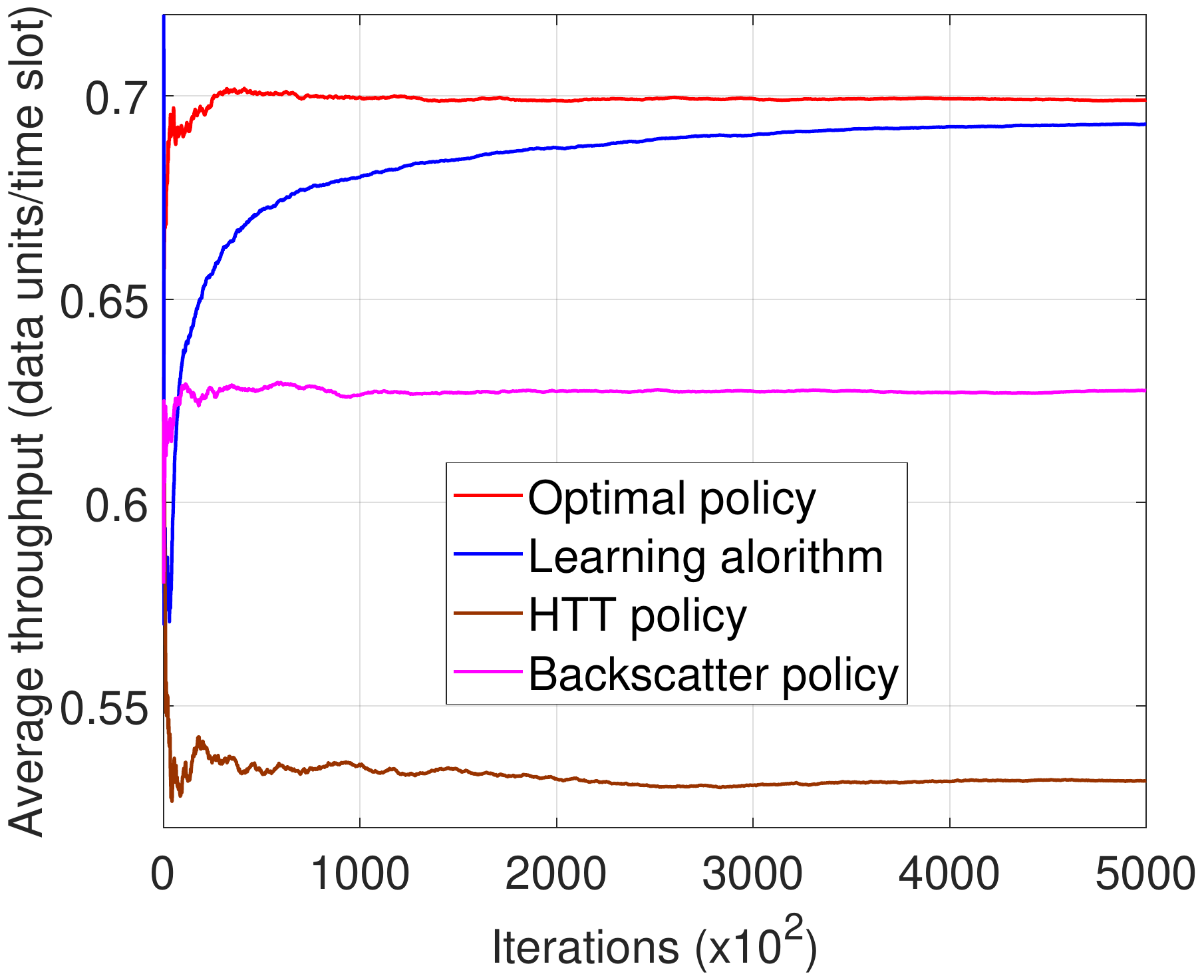}
		\caption{}
	\end{subfigure}
	\caption{Convergence of the learning algorithm in (a) the first $10^4$ iterations and (b) the first $5 \times 10^5$ iterations.} 
	\label{fig:convergence}
\end{figure*}

\subsubsection{Performance Comparison}
Next, we perform simulations to evaluate and compare performance of the proposed solutions, i.e., MDP optimization with complete information and the proposed online reinforcement learning algorithm (Algorithm~\ref{algorithm1}) with incomplete information, with three other policies, i.e., HTT, backscatter, and random policies, in terms of average throughput, delay, and blocking probability. 

\paragraph{Average throughput}
In Figs.~\ref{fig:throughput}(a) and (b), we show the average throughput of the ST obtained by the different policies when the idle channel probability and the packet arrival probability are varied, respectively. As observed in Fig.~\ref{fig:throughput}(a), as the idle channel probability increases, the average throughput of the HTT policy increases accordingly. This is due to the fact that when the incumbent channel is likely to be idle, the ST has more opportunity to transmit data from the data queue. However, when the idle channel probability is very high, i.e., $\eta \geq 0.6$, the average throughput obtained by the HTT policy will be reduced as the ST has less time to harvest energy, resulting in a low throughput. For the backscatter policy, since the ST only backscatters to transmit data, its performance will depend on the channel status. As a result, the average throughput of the ST in this case decreases as the idle channel probability increases. 

By switching among the actions of harvesting energy, backscattering, and active transmitting data, the optimal policy obtained from the aforementioned MDP-based optimization formulation achieves the highest throughput. Intuitively, when the idle channel probability is lower than 0.4, the ST will prefer the backscatter mode, and it will switch to HTT mode when the idle channel probability is higher than 0.4. We observe that the learning algorithm yields the throughput close to that of the optimal policy, and it is much higher than that of the other policies, e.g., about 17\% and 50\% higher than that of the random policy and the backscatter policy when $\eta = 0.7$, respectively.

\begin{figure*}[tbh]
	\centering
	\begin{subfigure}[b]{0.45\textwidth}
		\centering
		\includegraphics[scale=0.37]{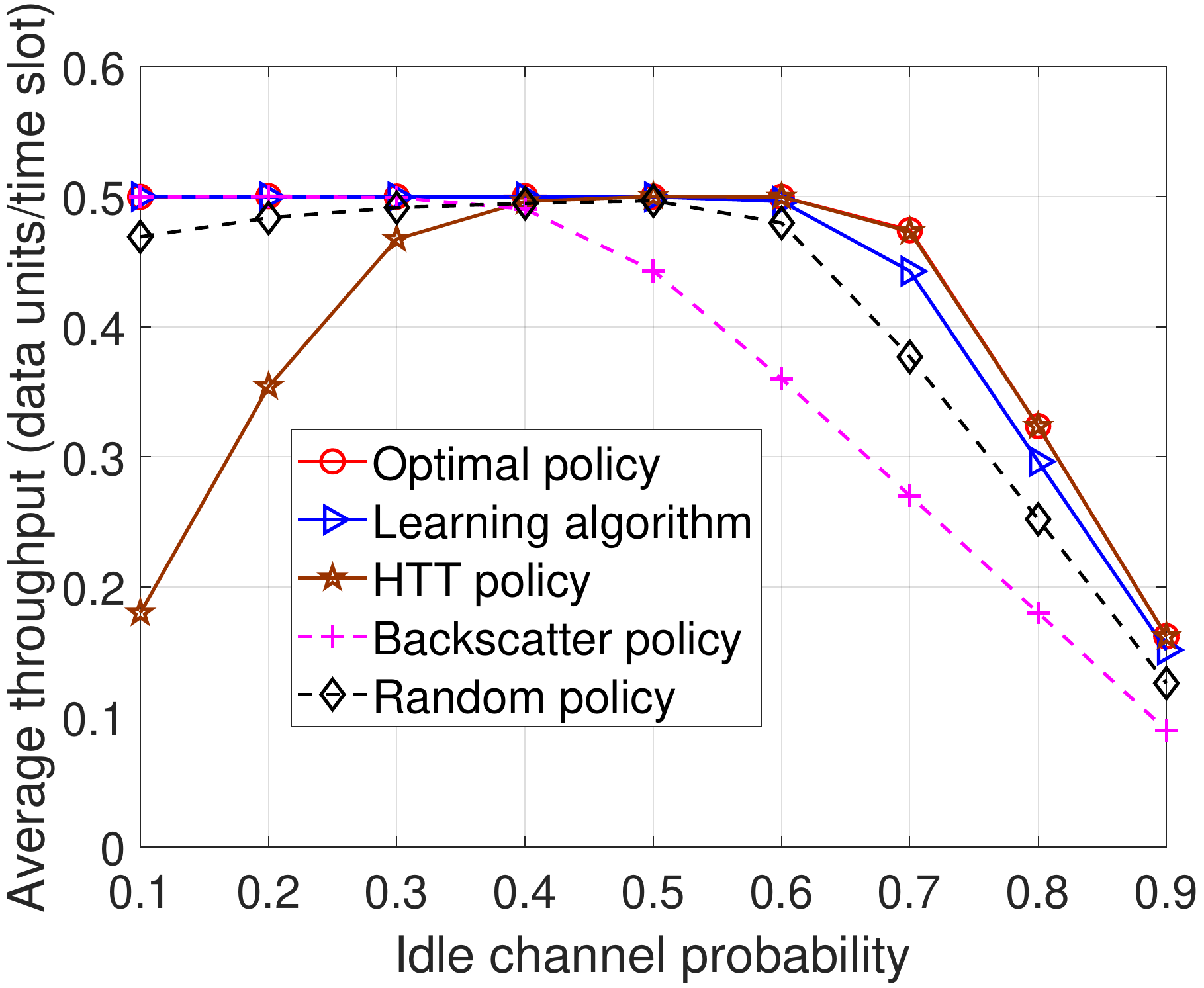}
		\caption{Idle channel probability is varied}
	\end{subfigure}%
	~ 
	\begin{subfigure}[b]{0.45\textwidth}
		\centering
		\includegraphics[scale=0.37]{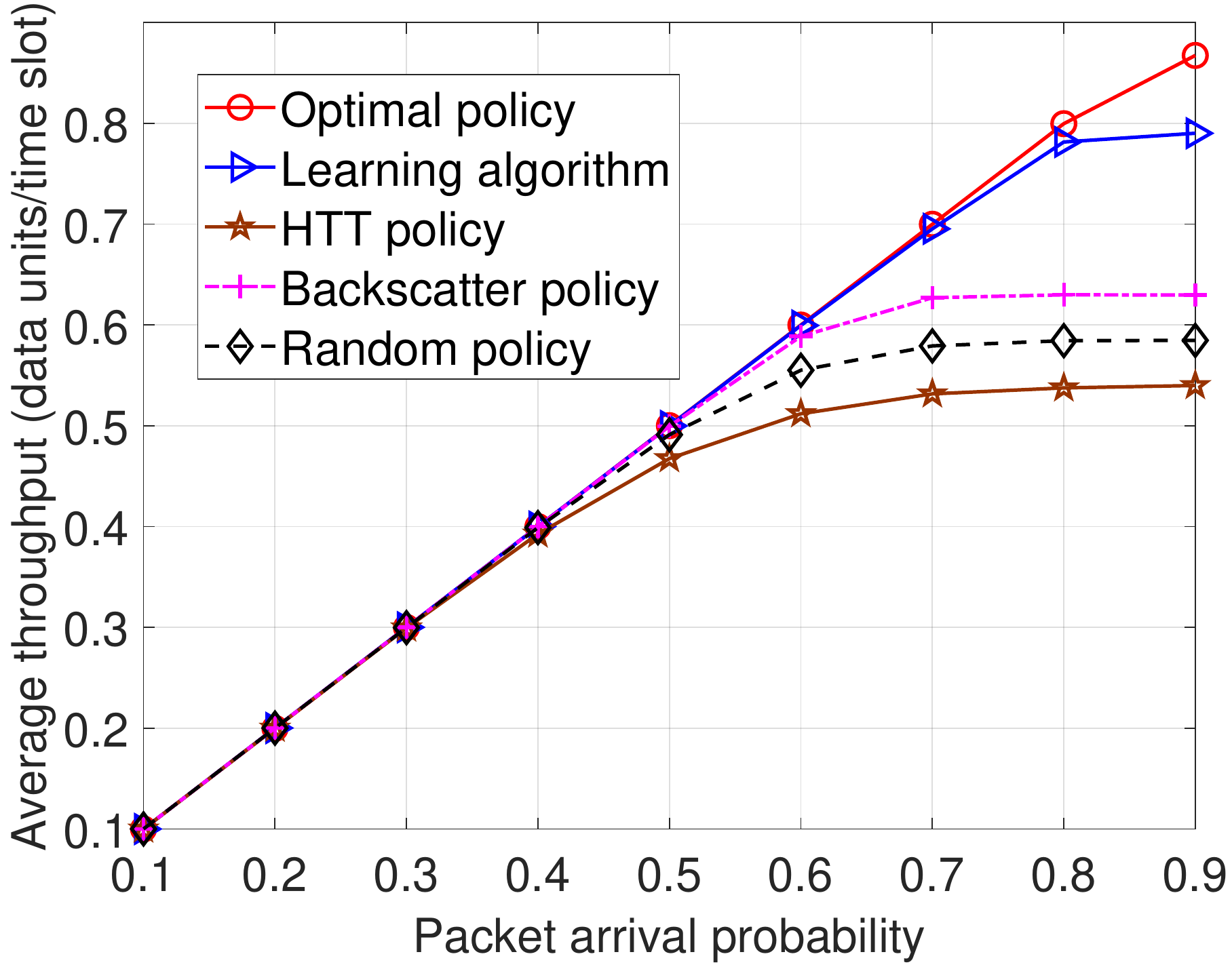}
		\caption{Packet arrival probability is varied}
	\end{subfigure}
	\caption{Average throughput of the secondary system.} 
	\label{fig:throughput}
\end{figure*}
Fig.~\ref{fig:throughput}(b) presents the throughput of the system when the packet arrival probability is varied. Clearly, when the packet arrival probability increases, the throughputs of all the policies increase. Under the small packet arrival probability, i.e., less than 0.4, all the policies yield almost the same throughput. This is due to the fact that when the number of data units in the data queue is very low, the ST has sufficient opportunity to transmit and/or backscatter its data as the probabilities that the incumbent channel is idle and busy are the same, i.e., $\eta = 0.5$. Similar to the case when the idle channel probability is varied, when the packet arrival probability is higher than 0.4, the optimal policy achieves the highest throughput followed by the learning algorithm. For example, when the packet arrival probability is 0.9, the throughput gain from using the learning algorithm can be up to about 50\% compared to the HTT policy.

\paragraph{Average delay and average blocking probability}
In Fig.~\ref{fig:packet} and Fig.~\ref{fig:blocking}, we examine the performance of the ST in terms of the average numbers of data units waiting in the data queue and the blocking probability, respectively. Specifically, in Fig.~\ref{fig:packet}(a) when the idle channel probability is lower than $0.6$, the average numbers of data units waiting in the data queue obtained by the optimal policy and the learning algorithm are very small. However, when the idle channel probability increases from $0.6$ to $0.9$, the average numbers of data units waiting in the data queue obtained by the two policies increase dramatically. Similar trends are observed in Fig.~\ref{fig:blocking}(a) for the blocking probability of the ST. 

In Fig.~\ref{fig:packet}(b) and Fig.~\ref{fig:blocking}(b), as the packet arrival probability increases from $0.1$ to $0.8$, the average number of data units waiting in the data queue and the blocking probability obtained by the optimal policy and the learning algorithm slightly increase. However, they increase significantly when the packet arrival probability reaches $0.9$. Note that in all cases, the average numbers of data units waiting in the data queue and the blocking probability obtained by the optimal policy and the learning algorithm always achieve the best performance (as little as 20\% of the other policies). This result is especially useful in controlling quality of service for the ST. 

\begin{figure*}[tbh]
	\centering
	\begin{subfigure}[b]{0.45\textwidth}
		\centering
		\includegraphics[scale=0.37]{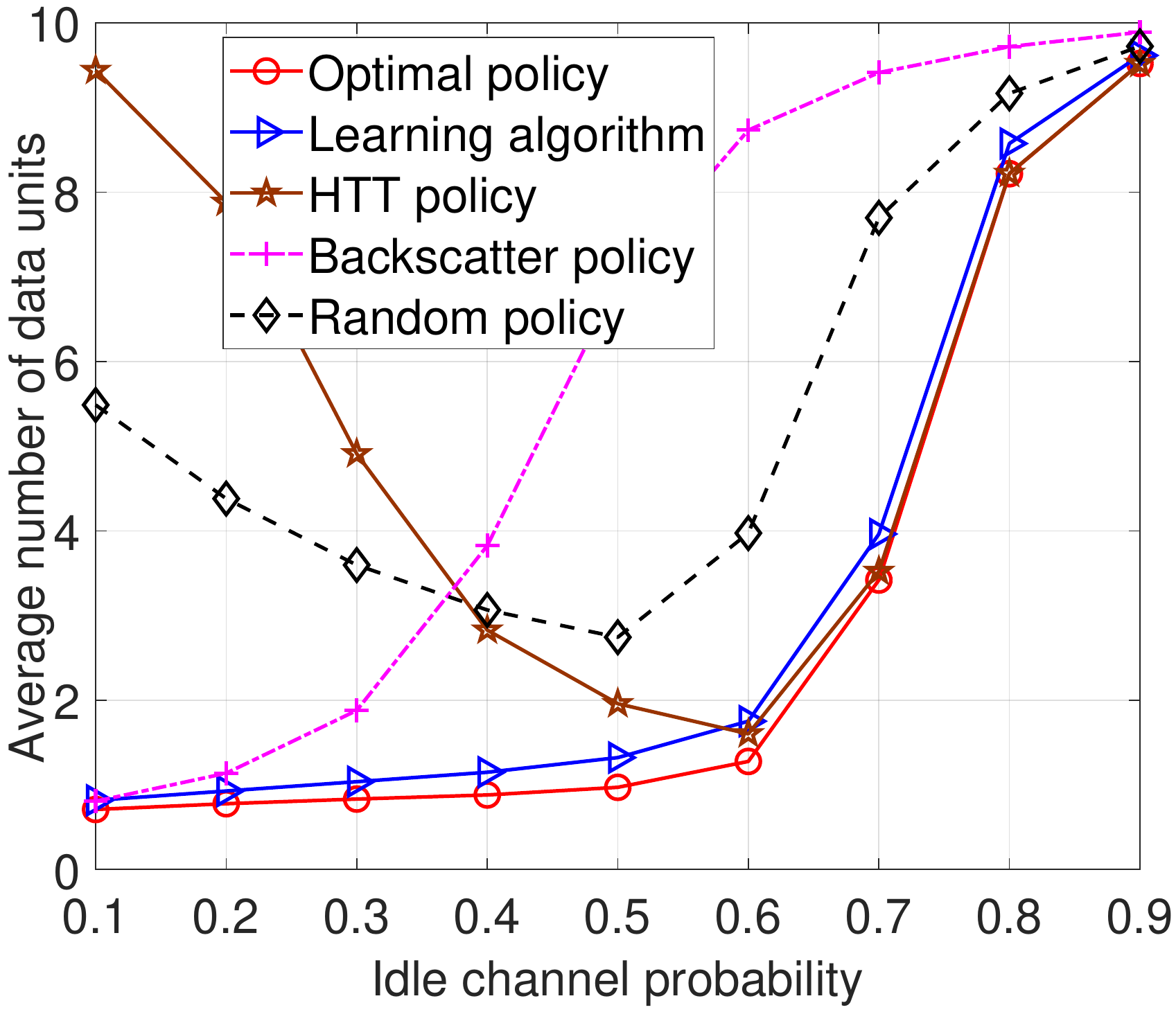}
		\caption{Idle channel probability is varied}
	\end{subfigure}%
	~ 
	\begin{subfigure}[b]{0.45\textwidth}
		\centering
		\includegraphics[scale=0.37]{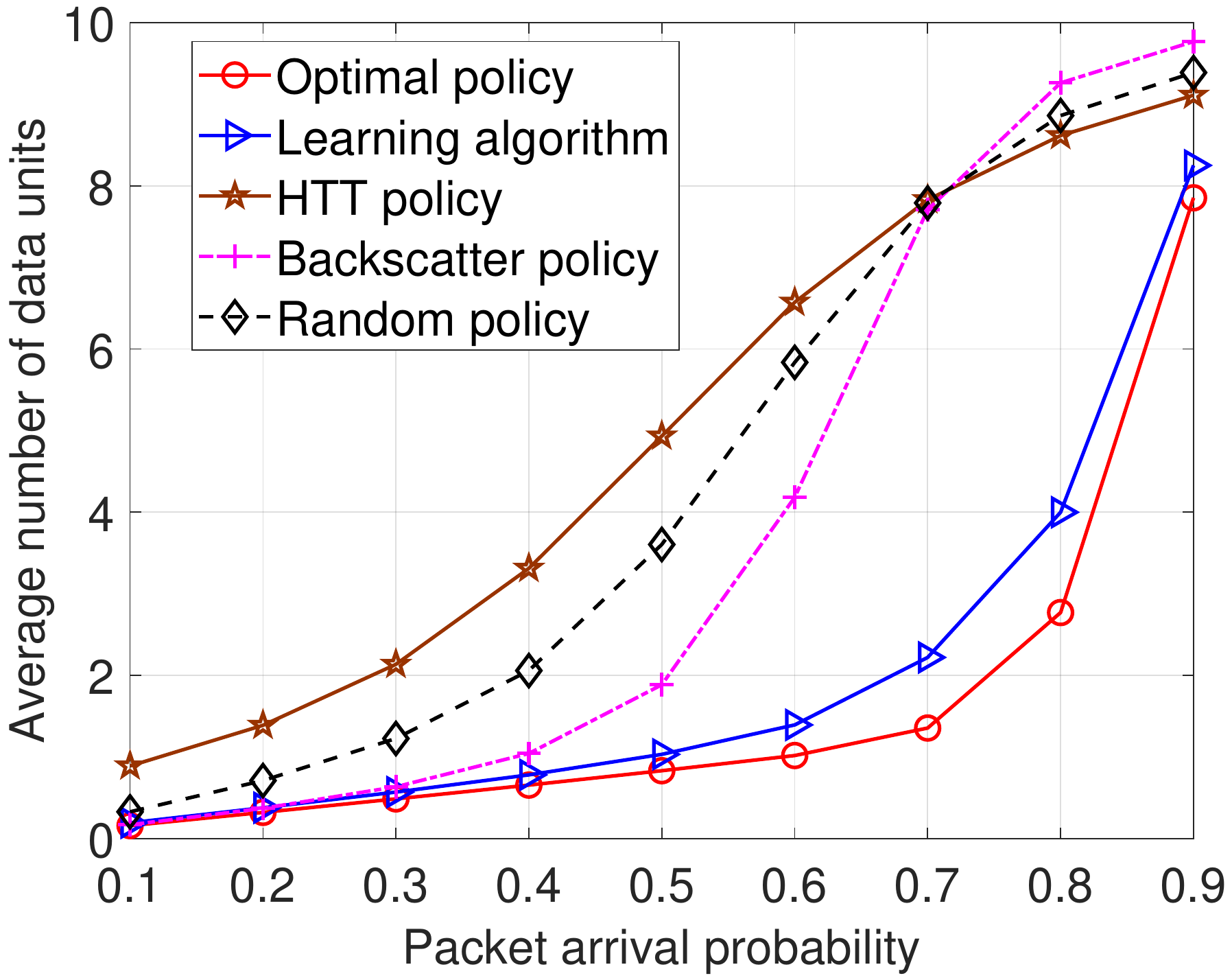}
		\caption{Packet arrival probability is varied}
	\end{subfigure}
	\caption{Average number of data units waiting in the data queue.} 
	\label{fig:packet}
\end{figure*}

\begin{figure*}[tbh]
	\centering
	\begin{subfigure}[b]{0.45\textwidth}
		\centering
		\includegraphics[scale=0.37]{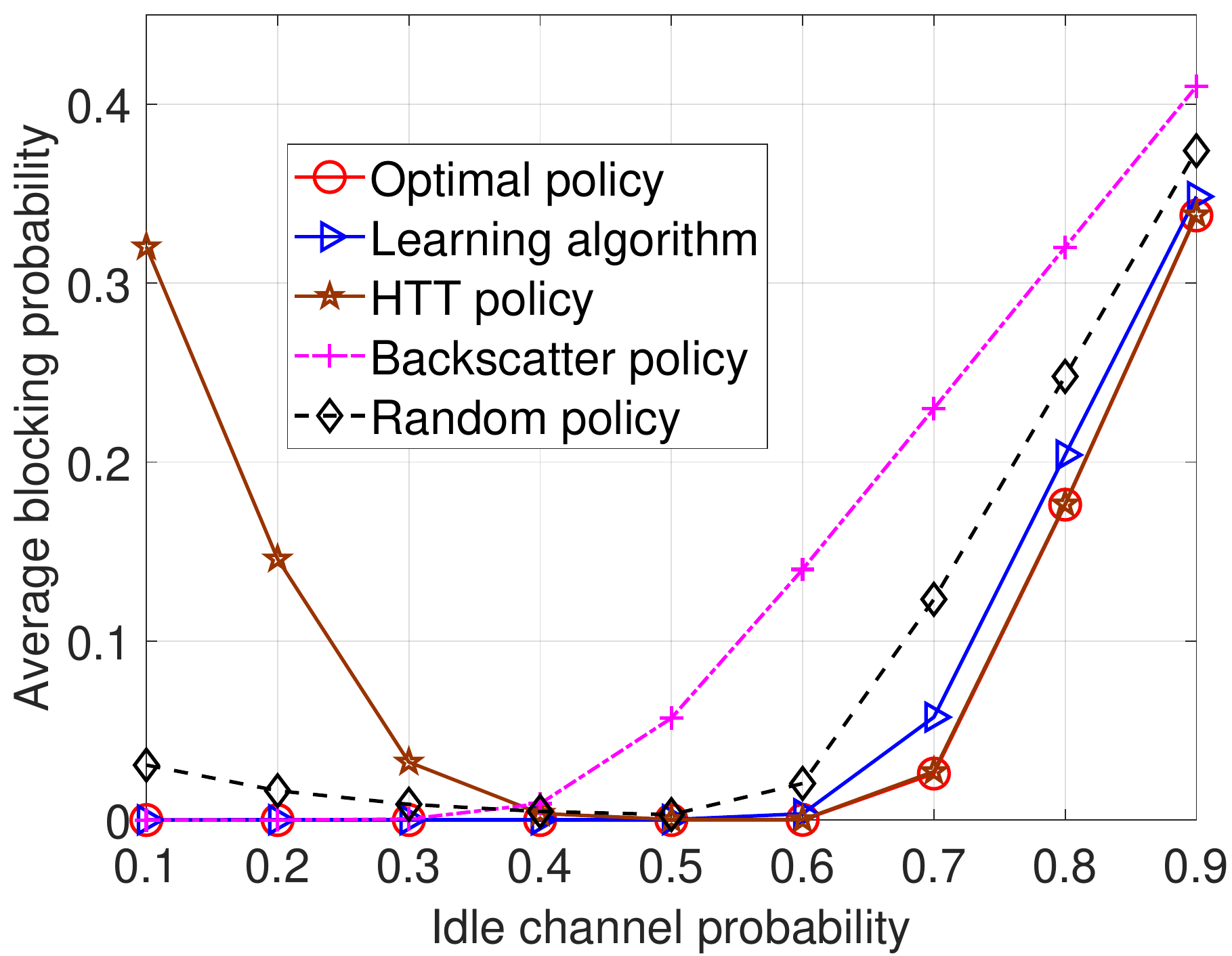}
		\caption{Idle channel probability is varied}
	\end{subfigure}%
	~ 
	\begin{subfigure}[b]{0.45\textwidth}
		\centering
		\includegraphics[scale=0.37]{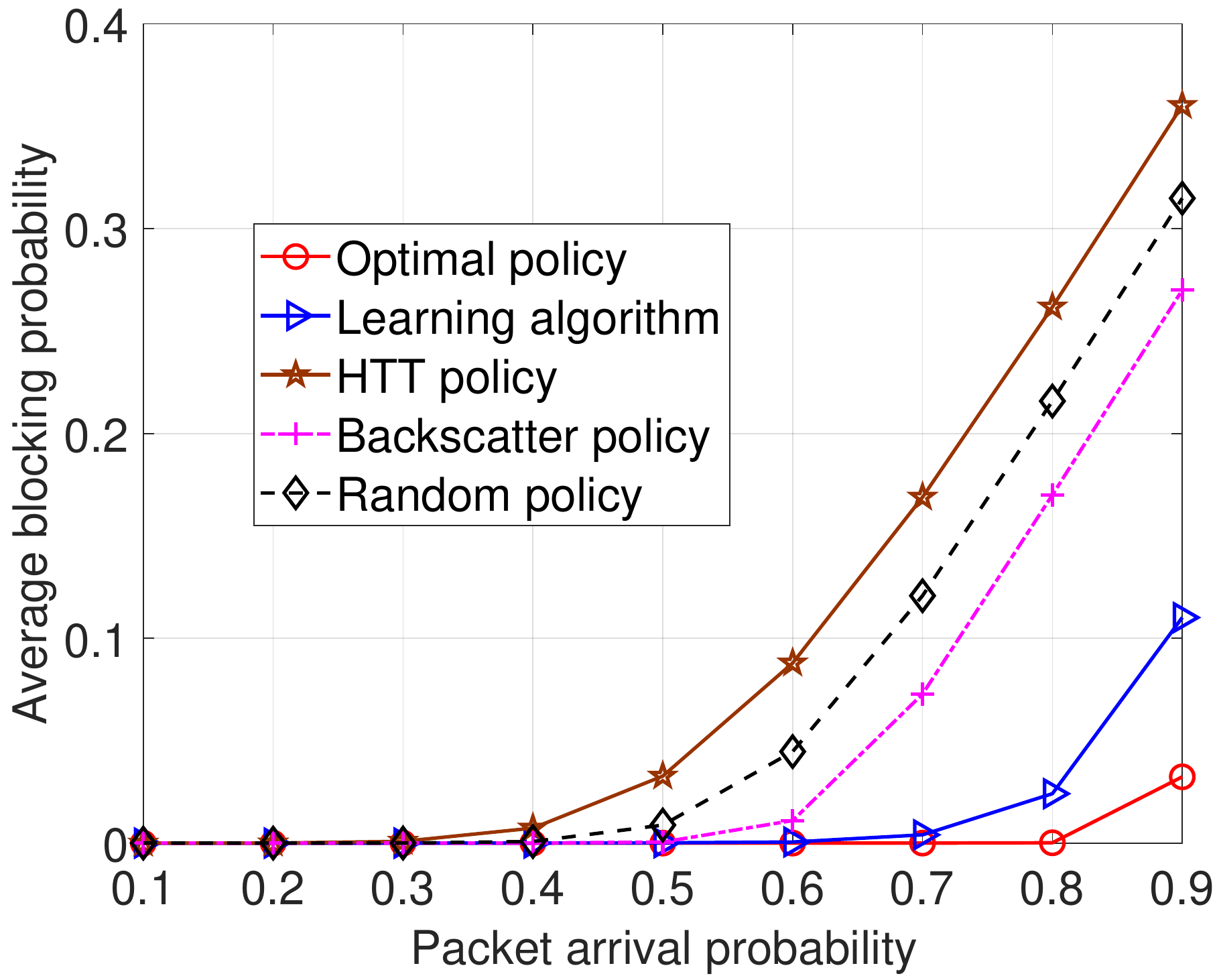}
		\caption{Packet arrival probability is varied}
	\end{subfigure}
	\caption{Average blocking probability.} 
	\label{fig:blocking}
\end{figure*}
\section{Conclusions}
\label{sec:conclusion}

In this paper, we have considered the DSA RF-powered ambient backscatter system in which the ST is equipped with RF-energy harvesting and ambient backscatter capabilities. In the system, the ST can harvest energy from incumbent signals or backscatter such signals to transmit data to its receiver when the incumbent channel is busy. To maximize the network performance under the dynamics of the environment and demands, the ST needs to choose the best action given its current state. We have introduced an MDP-based optimization framework to obtain the optimal policy for the ST. We have also developed a low-complexity online reinforcement learning algorithm that allows the ST to make optimal decisions when the complete environment parameters are not available. Through the numerical results, we have demonstrated that by using the proposed MDP optimization and the online reinforcement learning algorithm, the performance of the secondary system can be significantly improved compared with those of using HTT or backscatter individually. Moreover, the numerical results can provide insightful guidance for the ST to choose the best mode to operate.

\appendices
\section{The proof of Proposition~\ref{derivatives}}
\label{appendix:pro_derivatives}
Given state $s$, if the ST chooses to stay idle, i.e., $a=1$, or to harvest energy, i.e., $a=3$, the immediate throughput function $\mathcal{T}(s,a)=0$. Thus, we have
\begin{equation}
\begin{aligned}
{\mathcal{T}}_{\Theta} (s) & = \sum_{a \in \mathcal{A}} \chi_{\Theta}(s, a){\mathcal{T}}( s, a )
=\sum_{a \in \mathcal{A}} \frac{\exp\big(\theta_{ s,a }\big)} {\sum_{a' \in \mathcal{A}}\exp\big(\theta_{ s,a'  }\big)} \mathcal{T}( s, a )\\
&= \frac{\exp(\theta_{ s,2})\mathcal{T}(s,2) + \exp(\theta_{ s,4})\mathcal{T}(s,4)}{\exp(\theta_{ s,1})+\exp(\theta_{ s,2})+ \exp(\theta_{ s,3})+ \exp(\theta_{ s,4})}\\
&= \frac{\exp(\theta_{ s,2})\mathcal{T}(s,2) + \exp(\theta_{ s,4})\mathcal{T}(s,4)}{\mathcal{Q}},
\end{aligned}
\end{equation}
where $\mathcal{Q}=\exp(\theta_{ s,1})+\exp(\theta_{ s,2})+ \exp(\theta_{ s,3})+ \exp(\theta_{ s,4})$. Next, we derive the first derivative of the immediate throughput function ${\mathcal{T}}_{\Theta} (s)$ as $\nabla {\mathcal{T}}_{\Theta} (s)= \big[0,\frac{\partial {\mathcal{T}}_{\Theta} (s)}{\partial \theta_{ s,2}}, 0, \frac{\partial {\mathcal{T}}_{\Theta} (s)}{\partial \theta_{ s,4}}\big]$.
\begin{equation}
\begin{aligned}
\frac{\partial {\mathcal{T}}_{\Theta} (s)}{\partial \theta_{ s,2}} &= \frac{k \exp(\theta_{ s,2}) \mathcal{Q}-\frac{\partial \mathcal{Q}}{\partial \theta_{ s,2}}\exp(\theta_{ s,2})}{\mathcal{Q}^2} \mathcal{T}(s,2) - \frac{\frac{\partial \mathcal{Q}}{\partial \theta_{ s,2}}\exp(\theta_{ s,2})\exp(\theta_{ s,4})}{\mathcal{Q}^2}\mathcal{T}(s,4)\\
&= \frac{k \exp(\theta_{ s,2})\mathcal{Q}-k\exp(\theta_{ s,2})^2}{\mathcal{Q}^2}\mathcal{T}(s,2) - \frac{k \exp(\theta_{ s,2}) \exp(\theta_{ s,4})}{\mathcal{Q}^2}\mathcal{T}(s,4)\\
&=k\Big(\frac{\exp(\theta_{ s,2})}{\mathcal{Q}}-\big(\frac{\exp(\theta_{ s,2})}{\mathcal{Q}}\big)^2\Big)\mathcal{T}(s,2) - \frac{k \exp(\theta_{ s,2}) \exp(\theta_{ s,4})}{\mathcal{Q}^2}\mathcal{T}(s,4).
\end{aligned}
\end{equation}
Obviously, as $\exp(\theta_{ s,a}) > 0$ and $\theta$ is limited (see the proof of Theorem~\ref{prop2}), $\frac{\exp(\theta_{ s,2})}{\mathcal{Q}}$, $\big(\frac{\exp(\theta_{ s,2})}{\mathcal{Q}}\big)^2$, and $\frac{\exp(\theta_{ s,2}) \exp(\theta_{ s,4})}{\mathcal{Q}^2}$ are bounded. Thus, $\frac{\partial {\mathcal{T}}_{\Theta} (s)}{\partial \theta_{ s,2}}$ is bounded. Similarly, $\frac{\partial {\mathcal{T}}_{\Theta} (s)}{\partial \theta_{ s,4}}$ is also bounded. As a result, the first derivative of the immediate throughput function is bounded.

In the same way, we can derive the second derivative of the immediate throughput function ${\mathcal{T}}_{\Theta} (s)$ as $\nabla^2 {\mathcal{T}}_{\Theta} (s)= \big[0,\frac{\partial^2 {\mathcal{T}}_{\Theta} (s)}{\partial^2 \theta_{ s,2}}, 0, \frac{\partial^2 {\mathcal{T}}_{\Theta} (s)}{\partial^2 \theta_{ s,4}}\big]$.

\begin{equation}
\begin{aligned}
\frac{\partial^2 {\mathcal{T}}_{\Theta} (s)}{\partial^2 \theta_{ s,2}} &= k^2 \mathcal{T}(s,2)\Bigg(\frac{\exp(\theta_{ s,2})}{\mathcal{Q}}\Big(1-\frac{\exp(\theta_{ s,2})}{\mathcal{Q}}\Big)^2\Bigg)\\
&-k^2\mathcal{T}(s,4)\frac{\exp(\theta_{ s,2})}{\mathcal{Q}}\frac{\exp(\theta_{ s,4})}{\mathcal{Q}} + 2k^2\mathcal{T}(s,4)\big(\frac{\exp(\theta_{ s,2})}{\mathcal{Q}}\big)^2\frac{\exp(\theta_{ s,4})}{\mathcal{Q}}.
\end{aligned}
\end{equation}

Clearly, $\frac{\exp(\theta_{ s,2})}{\mathcal{Q}}$ and $\frac{\exp(\theta_{ s,4})}{\mathcal{Q}}$ are bounded. Thus, $\frac{\partial^2 {\mathcal{T}}_{\Theta} (s)}{\partial^2 \theta_{ s,2}}$ is bounded. Similarly, $\frac{\partial^2 {\mathcal{T}}_{\Theta} (s)}{\partial^2 \theta_{ s,4}}$ is also bounded. Therefore, the second derivative of the immediate throughput function is bounded. Similar with the immediate throughput function, the transition probability function $p_{\Theta}({s,s'})$ is also twice differentiable, and its first and second derivative are bounded.
\section{The proof of Theorem~\ref{prop_policy_gradient}}
\label{appendix:prop_policy_gradient}
This is to show how to calculate the gradient of the average throughput. In~(\ref{eq: balance equation}), with $\sum_{s \in \mathcal{S}} \pi_{\Theta}({s})  = 1$, we have $\sum_{s \in \mathcal{S}} \nabla \pi_{\Theta}({s})  = 0$. Recall that 
\begin{displaymath}
d(s, \Theta)={\mathcal{T}}_{\Theta}(s ) - \xi(\Theta) + \sum_{s \in \mathcal{S}} p_{\Theta}({s,s'}) d(s', \Theta), \phantom{5} \text{and} \phantom{5}  \xi(\Theta) = \sum_{s \in \mathcal{S}} \pi_{\Theta}({s}) {\mathcal{T}}_{\Theta} ( s ).
\end{displaymath}
Then, the gradient of $\xi(\Theta)$ is obtained as follows:
\begin{equation}
\label{xxxx0}
\begin{aligned}
\nabla \xi(\Theta)	 = &	\sum_{s \in \mathcal{S}} \pi_{\Theta}(s) \nabla \mathcal{T}_{\Theta}(s ) + \sum_{s \in \mathcal{S}} \nabla \pi_{\Theta}(s) \mathcal{T}_{\Theta}(s ) \\
= &	\sum_{s \in \mathcal{S}} \pi_{\Theta}(s) \nabla \mathcal{T}_{\Theta}(s ) +\sum_{s \in \mathcal{S}} \nabla \pi_{\Theta}(s)\big(\mathcal{T}_{\Theta}(s ) - \xi(\Theta) \big)		\\
= &	\sum_{s \in \mathcal{S}} \pi_{\Theta}(s) \nabla \mathcal{T}_{\Theta}(s ) +
\sum_{s \in \mathcal{S}} \nabla \pi_{\Theta}(s)\bigg(d(s, \Theta)		 -  \sum_{s \in \mathcal{S}} p_{\Theta}({s,s'}) d(s', \Theta) \bigg).
\end{aligned}
\end{equation}
We define
\begin{equation}
\label{appendix:eq:derivation}
\nabla \Big( \pi_{\Theta}(s) p_{\Theta}(s,s') \Big) = \nabla \pi_{\Theta}(s) p_{\Theta}(s,s') + \pi_{\Theta}(s) \nabla p_{\Theta}(s,s').
\end{equation}
and from~(\ref{eq: balance equation}), $\pi_{\Theta}(s')=\sum_{s \in \mathcal{S}} \pi_{\Theta}(s) p_{\Theta}(s,s')$. Then, equation (\ref{xxxx0}) can be expressed as follows:
\begin{equation}
\label{xxxx1}
\begin{aligned}
&\nabla \xi(\Theta)	 = \sum_{s \in \mathcal{S}} \pi_{\Theta}(s) \nabla \mathcal{T}_{\Theta}(s ) + \sum_{s \in \mathcal{S}} \nabla \pi_{\Theta}(s)\bigg(d(s, \Theta) 	-  \sum_{s \in \mathcal{S}} p_{\Theta}({s,s'}) d(s', \Theta) \bigg)	\\
= &\sum_{s \in \mathcal{S}} \pi_{\Theta}(s) \nabla \mathcal{T}_{\Theta}(s ) + \sum_{s \in \mathcal{S}} \nabla \pi_{\Theta}(s)d(s, \Theta) + \sum_{s,s' \in \mathcal{S}} \pi_{\Theta}(s) \nabla p_{\Theta}(s,s') d(s', \Theta) \\ 
&  - 	\sum_{s' \in \mathcal{S}} \nabla \Big( \sum_{s \in \mathcal{S}} \pi_{\Theta}(s) p_{\Theta}(s,s')  \Big) d(s', \Theta)		\\
= & \sum_{s \in \mathcal{S}} \pi_{\Theta}(s) \nabla \mathcal{T}_{\Theta}(s ) + \sum_{s \in \mathcal{S}} \nabla \pi_{\Theta}(s)d(s, \Theta)	+ \sum_{s,s' \in \mathcal{S}} \pi_{\Theta}(s) \nabla p_{\Theta}(s,s') d(s', \Theta) - \sum_{s' \in \mathcal{S}} \nabla \pi_{\Theta}(s') d(s', \Theta)	\\
= &\sum_{s \in \mathcal{S}} \pi_{\Theta}(s)\bigg( \nabla \mathcal{T}_{\Theta}(s ) + \sum_{s' \in \mathcal{S}} \nabla p_{\Theta}(s,s') d(s',\Theta)	\bigg).
\end{aligned}
\end{equation}
The proof is completed.
\section{The proof of Theorem~\ref{prop2}}
\label{appendix:prop2}
Let Proposition~\ref{recurrent_state}, Proposition~\ref{derivatives} hold, we prove the theorem in the following.

First, we reformulate the equations (\ref{theta_al1}) and (\ref{xi_al1}) in the specific form as follows:
\begin{equation}
\label{xxxx_2}
\begin{aligned}
&\Theta_{m+1} = \Theta_{m} + \rho_{m} \left( \sum_{k'=k_{m}}^{k_{m+1}-1} \Big( \sum_{k=k'}^{k_{m+1}-1}({\mathcal{T}} (s_{k}, a_{k}) - \widetilde{\xi}_{m}) \Big) \frac{\nabla \chi_{\Theta_{m}}(s_{k'},a_{k'})}{\chi_{\Theta_{m}}(s_{k'},a_{k'})} \right), \\
&\widetilde{\xi}_{m+1} = \widetilde{\xi}_{m} + \nu\rho_{m}\sum_{k'=k_{m}}^{k_{m+1}-1}({\mathcal{T}}(s_{k}, a_{k}) - \widetilde{\xi}_{m}).
\end{aligned}
\end{equation}
We define the vector $\mathbf{r}^{k_m}=\left[	\begin{array}{cc}	\Theta_m	&	\widetilde{\xi}_{m}	\end{array}	\right]^\top$, then~(\ref{xxxx_2}) becomes
\begin{equation}
\mathbf{r}^{k_{m+1}} = \mathbf{r}^{k_m} + \rho_{m} \mathbf{H}_m,
\end{equation}
where 
\begin{equation}
\mathbf{H}_m \!	= \!	\left[\!	\begin{array}{c}
\sum_{k'=k_{m}}^{k_{m+1}-1} \Big( \sum_{k=k'}^{k_{m+1}-1}({\mathcal{T}} (s_{k}, a_{k}) - \widetilde{\xi}_{m}) \Big) \frac{\nabla \chi_{\Theta_{m}}(s_{k'},a_{k'})}{\chi_{\Theta_{m}}(s_{k'},a_{k'})}	\\
\nu \sum_{k'=k_{m}}^{k_{m+1}-1}({\mathcal{T}}(s_{k}, a_{k}) - \widetilde{\xi}_{m})
\end{array}	\! \right].
\end{equation}
Denote $\mathscr{F} = \{ \Theta_0, \widetilde{\xi}_{0}, s_0,s_1,\ldots,s_m  \}$ as the history of the Algorithm~\ref{algorithm0}. Then, based on Proposition 2 in~\cite{Marbach2001}, we have
\begin{equation}
\begin{aligned}
\mathbb{E}[\mathbf{H}_m|\mathscr{F}_m]	& \!=\! \mathbf{h}_m	\!=\! \left[	\begin{array}{c}
\mathbb{E}_{\Theta}[T] \nabla \xi(\Theta) + \mathscr{V}(\Theta) \big( \xi(\Theta) - \widetilde{\xi}(\Theta) \big)	\\
\nu \mathbb{E}_{\Theta}[T] \big( \xi(\Theta) - \widetilde{\xi}(\Theta)	\big)
\end{array}	\right],
\end{aligned}
\end{equation}
where 
\begin{displaymath}
\mathscr{V}(\Theta)=\mathbb{E}_{\Theta}\Bigg[ \sum_{k'=k_{m+1}}^{k_{m+1}-1} \big( k_{m+1} - k'\big) \frac{\nabla \chi_{\Theta_{m}}(s_{k'},a_{k'})}{\chi_{\Theta_{m}}(s_{k'},a_{k'})} \Bigg].
\end{displaymath}
Recall that $T=\min\{k>0 | s_k = s^\dagger\}$ is the first future time that the algorithm visits the recurrent state $s^\dagger$. Therefore, the expression in (\ref{xxxx_2}) can be formulated as follows:
\begin{equation}
\mathbf{r}^{k_{m+1}} = \mathbf{r}^{k_m} + \rho_{m} \mathbf{h}_m + \varepsilon_m,
\end{equation}
where $\varepsilon_m = \rho (\mathbf{H}_m - \mathbf{h}_m)$ and note that $\mathbb{E}[\varepsilon_m|\mathscr{F}_m]=0$. As $\varepsilon_m$ and $\rho_m$ converge to zero almost surely, and $\mathbf{h}_{m}$ is bounded, we have
\begin{equation}
\lim_{m \rightarrow \infty} (\mathbf{r}^{k_{m+1}}-\mathbf{r}^{k_{m}}) = 0.
\end{equation}
After that, from Lemma 11 in~\cite{Marbach2001}, it is proved that $\xi(\Theta)$ and $\widetilde{\xi}(\Theta)$ converge to a common limit. Hence, the parameter vector $\Theta$ can be expressed as follows:
\begin{equation}
\label{eq:appendix_standard_form}
\Theta_{m+1} = \Theta_{m} + \rho_m \mathbb{E}_{\Theta_m}[T] (\nabla \xi(\Theta_m)+e_m) + \epsilon_m,
\end{equation}
where $\epsilon_m$ is a summable sequence and $e_m$ is an error term that converges to zero. As stated in~\cite{Bertsekas1999,Borkar2008}, (\ref{eq:appendix_standard_form}) is known as the gradient method with diminishing errors, and we can prove that $\nabla \xi(\Theta_m)$ converges to $0$, i.e., $\nabla_{\Theta} \xi(\Theta_{\infty})=0$.

\renewcommand{\baselinestretch}{1}

\end{document}